\documentclass[11pt]{article}

\usepackage[left=1in, right=1in, top=1in, bottom=1in]{geometry}

\usepackage[group-separator={,}]{siunitx}
\usepackage{amsmath,amsfonts,amssymb}
\usepackage{amsthm}
\usepackage{graphicx}
\usepackage{mathtools}
\usepackage{enumitem}
\usepackage{verbatim}
\usepackage{nccmath}
\usepackage{xspace}
\usepackage{xcolor}
\definecolor{MyBlue}{rgb}{0.12, 0.12, 0.76}
\usepackage{booktabs, multirow}
\usepackage{refcount}
\usepackage{caption}
\definecolor{MyBlue}{rgb}{0.12, 0.12, 0.76}
\usepackage[colorlinks,allcolors=MyBlue]{hyperref}
\usepackage[numbers]{natbib}
\usepackage{algorithm}
\usepackage{algpseudocode}
\usepackage{algorithmicx}
\let\oldReturn\Return
\renewcommand{\Return}{\State\oldReturn}
\algtext*{EndWhile}% Remove "end while" text
\algtext*{EndIf}% Remove "end if" text
\algtext*{EndForAll}% Remove "end for" text
\algtext*{EndFor}% Remove "end for" text
\algtext*{EndFunction}% Remove "end function" text
\usepackage{pgf}
\usepackage{tikz}
\usetikzlibrary{shadows,arrows,decorations,decorations.shapes,backgrounds,shapes,snakes,automata,fit,petri,shapes.multipart,calc,positioning,shapes.geometric,graphs,graphs.standard}
\usepackage{tikz-qtree}
\DeclareMathOperator*{\argmax}{arg\,max}

\makeatletter % define protocol environment
\newenvironment{protocol}[1][htb]{%
    \renewcommand{\ALG@name}{Protocol}% change algorithm name to "protocol"
   \begin{algorithm}[#1]%
  }{\end{algorithm}}
\makeatother
\newtheorem{condition}{Condition}[section]
\newtheorem{lemma}{Lemma}[section] 
\newtheorem{theorem}{Theorem}[section]

\newtheorem{definition}{Definition}[section]
\newcommand\protocolspacing{\vspace{.05 in}} % spacing under protocol caption
\newcommand\valsize{v^{size}}
\newcommand{\T}{\mathcal{T}}
\newcommand\cals{\mathcal{S}}
\newcommand\bfc{\mathbf{c}}
\newcommand\properties{$\{\text{EF, Prop}\}$\xspace}
\newcommand\equ{\textsc{Equality}\xspace} % xspace is solving a spacing issue
\newcommand\disj{\textsc{Disjointness}\xspace}
\newcommand\fd{\textsc{Fair Division}\xspace}
\newcommand\overi{\overline{i}}
\newcommand\mss{M\backslash S}

\makeatletter
\newcommand{\thickhline}{%
    \noalign {\ifnum 0=`}\fi \hrule height 1.4pt
    \futurelet \reserved@a \@xhline
}
\newcolumntype{"}{@{\hskip\tabcolsep\vrule width 1.4pt\hskip\tabcolsep}}
\makeatother

\usepackage{thmtools}
\usepackage{thm-restate}

\begin{document}

\title{Communication Complexity of Discrete Fair Division}  
\author{Benjamin Plaut \and Tim Roughgarden}
\date{\{bplaut, tim\}@cs.stanford.edu\\ Stanford University}

\maketitle

\begin{abstract}
We initiate the study of the communication complexity of fair
division with indivisible goods. We focus on some of the most well-studied fairness notions (envy-freeness,
proportionality, and approximations thereof) and valuation classes
(submodular, subadditive and unrestricted). Within these parameters, our results completely resolve whether the communication complexity of
computing a fair allocation (or determining that none exist) is
polynomial or exponential (in the number of goods), for every combination of fairness notion, valuation class,
and number of players, for both deterministic and randomized protocols.
\end{abstract}

%\newpage
%\clearpage
%\setcounter{page}{1} % start from page 1 on the second page

% Main paper

\section{Introduction}\label{sec:intro}

Fair division studies the problem of distributing resources among competing
players in a ``fair" way, where each player has equal claim to the
resources. There are many different notions of fairness, with the two
most prominent being envy-freeness and proportionality. An allocation
is \emph{envy-free} (EF) if each player's value for her own bundle is
at least as much as her value for any other player's bundle. An
allocation is \emph{proportional} (Prop) if each player's value for
her bundle is at least $1/n$ of her value for the entire set of items,
where $n$ is the number of players.

In \emph{discrete} fair division, the resources consist of indivisible
items: each item must go to a single player and cannot be split among
players. Unfortunately, neither envy-freeness nor proportionality can
be guaranteed in this setting. Consider two players and a single item:
one must receive the item while the other receives nothing, so the
allocation is neither envy-free nor proportional. We study the problem
of finding an envy-free (or proportional) allocation, or showing that
none exists.
 
We also consider approximate versions of these properties: for
$c \in [0,1]$, an allocation is $c$-EF if each player's value for her
own bundle is at least $c$ times her value for any other player's bundle,
and an allocation is $c$-Prop if each player's value for her bundle is at least $c/n$ of her value for the entire set of items. Thus 1-EF and 1-Prop are
standard envy-freeness and proportionality, respectively. The same
counterexample of two players and a single item shows that these
approximate properties also cannot be guaranteed for any $c > 0$.\footnote{We generally assume that $c > 0$, since every allocation is both 0-EF and 0-Prop.}

From a computational complexity viewpoint,
this problem is hard 
even when player valuations are \emph{additive}, meaning that a
player's value for a set of items is the sum of her values for the
individual items. For two players with identical additive valuations,
determining whether a 1-EF or 1-Prop exists is NP-hard, via a simple
reduction from the partition problem~\cite{Bouveret08:Efficiency}.

It is arguably even more natural to study fair division from a
communication complexity perspective, where there is no centralized
authority and each player initially knows only her own preferences.

When players have \emph{combinatorial valuations}, their values for a
bundle cannot just be decomposed into their values for the individual
items.\footnote{An increasing amount of
  research in fair division considers such combinatorial valuations
  (e.g.~\cite{Plaut18:Almost,Barman17:Approximation, Ghodsi17:Fair}).} In particular, for $m$ items, a combinatorial valuation may
contain $2^m$ different values.
The primary question is to determine whether players need to exchange
an exponential amount of information to compute a fair allocation, or
whether the problem can be solved using only polynomial
communication.  This question has not been studied previously, despite
the rich literature on communication complexity in combinatorial
auctions
(e.g.~\cite{Nisan06:Communication,Nisan02:Set,Dobzinski10:Approximate}).

Our paper can also be thought of as formally studying the difficulty of eliciting different classes of valuations from a fair division standpoint. Additive valuations are typically used in practice (for example on the non-profit website Spliddit~\cite{Goldman15:Spliddit}) because each player need only report one value for each item to specify the entire valuation. Richer combinatorial valuations allow for more expressiveness, but may be more difficult to elicit. Our work formally studies the tradeoffs between these factors.

%%%%%%%%%%%%%%%%%%%%%%%%%%%%%%%%%%%%%
%%%%%%%%%%%%%%%%%%%%%%%%%%%%%%%%%%%%%
%%%%%%%%%%%%%%%%%%%%%%%%%%%%%%%%%%%%%
% OUR RESULTS

\subsection{Our results}\label{sec:contribution}

We study the following question: ``Given $n$ players and $m$ items, a fairness property $P \in$ \properties, and a constant $c \in [0,1]$, how much communication is required to either find a $c$-$P$ allocation, or show that none exists?"\footnote{We only consider a single $c$-$P$ property at a time: we do not consider satisfying envy-freeness and proportionality simultaneously. For subadditive valuations, $c$-EF implies $c$-Prop, but $c$-EF and $c$-Prop are incomparable for general valuations.} We are primarily interested in whether this can be done with communication polynomial in $m$. The answer to this question will depend on $n$, $P$, and $c$. We also consider when player valuations are restricted to be submodular or subadditive, as well as deterministic vs. randomized protocols.

All in all, we give a full characterization of the communication complexity for every combination of the following five parameters:

\begin{enumerate}
\item Number of players $n$
\item Valuation class: submodular, subadditive, or general
\item Each $P \in$ \properties
\item Every constant $c \in [0,1]$
\item Deterministic or randomized communication complexity
\end{enumerate}

%%%%%%%%%%%%%%%%%%%%%%%%%%%%%%%%%%%%%%%
\subsubsection{The importance of the two-player setting}

One of our results (Section~\ref{sec:n>2}) shows that there is no hope
for a polynomial communication protocol for more than two players:
exponential communication is required for every $n>2$, for either
$P \in$ \properties, for any $c > 0$, even for submodular valuations,
and even for randomized protocols.  
%Although this is disheartening,
The (very important) two-player case is surprisingly rich, however, with
multiple phenomena occurring across different valuation classes and
constants $c$. The results for two players in the deterministic
setting are summarized in Table~\ref{tbl:results}. It is also
surprising that there is such a chasm between the two- and
three-player cases; for example, there is no analogous chasm for
maximizing the social welfare in combinatorial auctions.

Furthermore, in contrast to combinatorial auctions, the two-player
setting is fundamental in fair division. Indeed, the first known
mention of fair division is in the Bible, when Abraham and Lot use the
cut-and-choose method to divide a piece of land. In modern day, one of
the primary applications of fair division for indivisible items is
divorce settlements, which is fundamentally a two-player setting. Fair
Outcomes Inc.\footnote{http://fairoutcomes.com}, a commercial fair
division website, only allows for two players. Other applications of
fair division, such as dividing an inheritance and international
border disputes, are also often two player settings. Unless otherwise
mentioned, we assume that $n=2$ throughout the paper.

%%%%%%%%%%%%%%%%%%%%%%%%%%%%%%%%%%%%%%%
\subsubsection{Submodular valuations}\label{sec:intro-submod}

We first consider submodular valuations in the deterministic setting
(for $n=2$). We show that full proportionality (1-Prop) requires only polynomial communication (Theorem~\ref{thm:prop-n=2-submod-upper}), whereas full envy-freeness
requires exponential communication
(Theorem~\ref{thm:EF-n=2-submod-lower}), exhibiting an interesting
difference between the two properties.

The hardness result for 1-EF leaves open the intriguing possibility of
a polynomial-communication approximation scheme (PAS):\footnote{This is the same idea as a polynomial-time approximation scheme (PTAS), but here we are interested in communication, not time.} for any fixed $c < 1$,
is communication cost polynomial in $m$ sufficient? 
As one of our main results, we prove
%Happily, 
that this is
indeed the case, and we prove it using a reduction to a type of graph we call the ``minimal bundle graph'' (Theorem~\ref{thm:EF-n=2-submod-upper}). This is our most technically involved argument.

The communication cost of this protocol
  exponential in $\frac{1}{1-c}$, and so this PAS is not a \emph{fully
    polynomial-communication approximation scheme} (FPAS), which would require
  polynomial dependence on $\frac{1}{1-c}$. Our lower bound for 1-EF
  (Theorem~\ref{thm:EF-n=2-submod-lower}) rules out an
  FPAS, so our results are still tight.

%%%%%%%%%%%%%%%%%%%%%%%%%%%%%%%%%%%%%%%
\subsubsection{Subadditive valuations}

The story is different for subadditive valuations, which are treated
in Section~\ref{sec:subadd}. We show that only polynomial
communication is required for $c$-EF when $c \leq 1/2$
(Theorem~\ref{thm:EF-n=2-subadd-upper}) and for $c$-Prop when
$c \leq 2/3$ (Theorem~\ref{thm:prop-n=2-subadd-upper}). Interestingly, the constants $1/2$ and $2/3$ turn out to be tight: we
show that exponential communication is required for $c$-EF for every
constant $c > 1/2$ (Theorem~\ref{thm:EF-n=2-subadd-lower}) and for $c$-Prop for
every constant $c > 2/3$ (Theorem~\ref{thm:prop-n=2-subadd-lower}).  This
establishes another interesting difference between the two fairness
notions.

%%%%%%%%%%%%%%%%%%%%%%%%%%%%%%%%%%%%%%%

\subsubsection{General valuations}

The story is again different for general valuations, which we consider
in Section~\ref{sec:general}. In the deterministic setting, $c$-EF and $c$-Prop each require exponential communication for every $c > 0$ (Theorems~\ref{thm:EF-n=2-general-lower} and \ref{thm:prop-n=2-rand-lower}). This resolves the deterministic setting.

\subsubsection{Randomized communication complexity}
The $c$-Prop lower bound for general valuations also holds in the randomized setting for any $c > 0$. However, $c$-EF admits an efficient randomized protocol for any $c \leq 1$ and general (and hence also subadditive and submodular) valuations. This randomized protocol is based on a reduction to the \equ problem (testing whether two bit strings are identical), which is known to have an efficient randomized protocol. Our randomized protocol for $c$-EF also carries over to $c$-Prop for any $c \leq 1$ in the special case of subadditive (and hence also submodular)
valuations. This resolves the randomized setting.

Finally, we briefly consider the maximin share property in Section~\ref{sec:maximin}, and prove exponential lower bounds in that setting as well.
 
%%%%%%%%%%%%%%%%%%%%%%%%%%%%%%%%%%%%%
\renewcommand{\arraystretch}{1.6}
\begin{table}
\resizebox{6.45 in}{!}{ %! means scale proportionately
% DETERMINISTIC n = 2 RESULTS
\begin{tabular}{c|c|c|c|c} 
\multicolumn{1}{c}{}  & \multicolumn{2}{c|}{$c$-EF (deterministic)} &
                                                                      \multicolumn{2}{c}{$c$-Prop (deterministic)}\\
  \cline{2-5}
\multicolumn{1}{c}{}  & easy when & hard when & easy when & hard when\\
  \thickhline
general valuations & 
never
& $c > 0$\ (Thm.~\ref{thm:EF-n=2-general-lower}) 
	& 
never
& $c > 0$\ (Thm.~\ref{thm:prop-n=2-rand-lower})\\
\hline
subadditive valuations & $c \leq 1/2$\ (Thm.~\ref{thm:EF-n=2-subadd-upper}) 
	& $c > 1/2$\ (Thm.~\ref{thm:EF-n=2-subadd-lower}) 
	& $c \leq 2/3$\ (Thm.~\ref{thm:prop-n=2-subadd-upper}) 
	& $c > 2/3$\ (Thm.~\ref{thm:prop-n=2-subadd-lower})\\
\hline
submodular valuations & $c < 1$\ (Thm.~\ref{thm:EF-n=2-submod-upper}) 
	& $c = 1$\ (Thm.~\ref{thm:EF-n=2-submod-lower}) & $c \leq 1$\
                                                          (Thm.~\ref{thm:prop-n=2-submod-upper})
                                  & 
%---
never\\   
 \end{tabular}
 }
 \vspace{.1 in}
 \caption{A summary of our results for the two-player deterministic
   setting. For both $c$-EF and $c$-Prop, we characterize exactly when
   the problem is easy (i.e., can be solved with communication polynomial in
   the number of items) and
   hard (i.e., requires exponential communication). We note that the protocol for Theorem~\ref{thm:EF-n=2-submod-upper} has communication cost exponential in $\frac{1}{1-c}$, and the corresponding lower bound (Theorem~\ref{thm:EF-n=2-submod-lower}) rules out a protocol with communication cost polynomial in $\frac{1}{1-c}$. See Section~\ref{sec:intro-submod} for additional discussion.
}
 \label{tbl:results}
\end{table}
\renewcommand{\arraystretch}{1}

%%%%%%%%%%%%%%%%%%%%%%%%%%%%%%%%%%%%%
%%%%%%%%%%%%%%%%%%%%%%%%%%%%%%%%%%%%%
% DESCRIPTION OF PROTOCOLS

\subsection{Ideas behind our protocols}\label{sec:techniques}

Since the problem is always hard when $n>2$, all of our upper bounds are in
the two-player setting. All of our positive results require the following
condition: for any partition of the items into two bundles $A_1$ and $A_2$, each player must be happy with at least one of $A_1$ and $A_2$. This is
always true for envy-freeness: a player is always happy with whichever
of $A_1$ and $A_2$ she has maximum value for (she could be happy with
both bundles if they have equal value to her). This is not satisfied
for proportionality in general, for example if a player has value zero
for each of $A_1$ and $A_2$, but positive value for $A_1 \cup A_2$. However, it is satisfied for subadditive
valuations.

All of our deterministic protocols have the same first step: if there
is any allocation where player 1 would be happy to receive either
bundle, she specifies that allocation to player 2, and player 2
selects her preferred bundle. Player 2 is guaranteed to be happy with
at least one of the bundles by the above condition, and player 1 is
happy with either bundle in this allocation, so she is happy as well.

The key to the analysis is what happens when there is no allocation
such that player 1 is happy with either bundle. It will turn out that the
absence of such an allocation implies certain structure in the
valuations. The exact structure, and the way the structure is
exploited, depends on the setting (valuation class, property $P$, and
constant $c$).

For example, consider the case of subadditive valuations and
$\frac{1}{2}$-EF. We show that if there is no
allocation where player 1 is happy with either bundle, then there must exist a
single item that player 1 values more than the rest of the items
combined. Then player 1 can simply specify that item to player 2. If
player 2 is happy with the rest of the items, we have found a
satisfactory allocation. Otherwise, there is no satisfactory
allocation, since player 1 and player 2 both care about that
particular item more than the rest of the items combined.

Furthermore, this protocol gives an additional guarantee. If a $c$-$P$
allocation is not returned, the protocol will return the fairest
allocation possible, i.e., a $c'$-$P$ allocation where no allocation
is $c''$-$P$ for any $c'' > c'$. For brevity, we will use $c^*$ to refer to the maximum $c'$ such that a $c'$-$P$ allocation exists.\footnote{It is possible that $c^*=0$ (for example, in the case of two
  players and one item), but our protocol at least certifies that this
  is the best possible.}
%\footnote{For consistency, we always use $c$ to represent the constant for which our true goal is to find a $c$-$P$ allocation, and use $c'$ to represent the best achievable approximation if a $c$-$P$ allocation does not exist. We use $c''$ only in the context of ``no allocation is $c''$-$P$ for any $c'' > c'$".} 
If player 2 determines that a $c$-$P$
allocation does not exist, then there is a single item $g$ that both
players care about more than all of the other items together. One
player will have to not receive item $g$, and the protocol gives $g$
to the player who will be most unhappy otherwise. This yields a $c^*$-$P$ allocation. In fact, all of our deterministic protocols give this guarantee, although slightly more work is required to achieve it in other settings.

\subsubsection{Minimal bundles}

The reasoning described above is actually a special case of analyzing what we call \emph{minimal bundles}. We say that a bundle is minimal for some player if that player is happy with the bundle, but is not happy with any strict subset of that bundle.\footnote{A similar notion of ``minimal bundles" features prominently in \cite{Brams2012:Undercut}.} The minimal bundles represent the most a player is willing to compromise. If a player does not receive one of her minimal bundles (or a superset thereof), she cannot be happy, by definition. On the other hand, if a player receives one of her minimal bundles (or a superset thereof), she is guaranteed to be happy.\footnote{We assume monotonicity: adding items to a bundle cannot decrease its value.} Thus it is both necessary and sufficient for each player to receive one of her minimal bundles (or a superset thereof). By this reasoning, it is sufficient for player 1 to specify all of her
minimal bundles to player 2: player 2 can then determine if there is an allocation which satisfies her (player 2), while still giving player 1 one of player 1's minimal
bundles. 

The general Minimal Bundle Protocol (Protocol~\ref{pro:minimal-upper}) is as follows. If there is an allocation where player 1 is happy with either bundle, she specifies that allocation to player 2, and we are done. Otherwise, player 1 specifies all of her minimal bundles to player 2, who searches for a satisfactory allocation. If player 2 fails to find one, she declares that no satisfactory allocation exists. There is a final step that is used to guarantee that a $c^*$-$P$ allocation is returned if no $c$-$P$ allocation is found; this will be described later.

The key is proving that the number of minimal bundles is polynomial in $m$, and this analysis varies based on the context. For example, for subadditive valuations and $\frac{1}{2}$-EF, we discussed above how if there is no allocation where player 1 is happy with either bundle, there must be a single item $g$ that she values more than all of the other items together. This means that player 1 has a single minimal bundle: $\{g\}$.

We also use the protocol to give a PAS for EF in the submodular
setting: we show that for every fixed $c < 1$, the number of minimal
bundles is at most $2(m+1)^{\frac{8}{1-c}}$, and thus the protocol
uses polynomial communication for any fixed $c$. The analysis
for this case is technically involved and involves constructing what we
call the ``minimal bundle graph" for player 1's valuation. The vertices in this graph are the minimal bundles, and two vertices share an edge if the corresponding bundles overlap by exactly one item (it will be impossible for two minimal bundles to overlap by more than one item). For some of these edges, moving the overlapping item between bundles will cause a large change in value: these special edges will play an important role. We will show that the only way to have a large number of minimal bundles is for there to be a large number of these special edges, but submodularity will imply an upper bound on how many special edges can be incident on a single vertex, and hence an upper bound on the total number of special edges.

The Minimal Bundle Protocol is correct for any valuation class, property $P$, or constant $c$. However, in some contexts, the number of minimal bundles may be exponential. Our lower bound constructions all involve valuations with an exponential number of minimal bundles.

%%%%%%%%%%%%%%%%%%%%%%%%%%%%%%%%%%%%%
%%%%%%%%%%%%%%%%%%%%%%%%%%%%%%%%%%%%%
%%%%%%%%%%%%%%%%%%%%%%%%%%%%%%%%%%%%%
% RELATED WORK

\subsection{Related work}
Fair division has a long history, and a full survey of this field is outside of the scope of this paper: see e.g. \cite{Brandt16:Handbook, Brams96:Cake, Moulin03:Fair} for further background.

There are several approaches for handling the fundamental asymmetry of
indivisible items, where neither envy-freeness or proportionality can be guaranteed. One natural question is whether there are other compelling properties that can be guaranteed~\cite{Plaut18:Almost, Caragiannis16:Nash, Budish16:Course, Lipton04:Approximately}. Another possibility is to allow for randomized allocations and search for allocations that are fair in expectation~\cite{Bogomolnaia04:Random, Budish13:Designing}.

%It is also possible to only consider the players' ordinal rankings of the items~\cite{Aziz15:Fair, Brams17:Maximin}: certain tasks become easier under this assumption, but important information is arguably lost as well.

Although EF and Prop allocations do not always exist, they often do. For example, \cite{Dickerson14:Computational} showed that when the number of items is at least a logarithmic factor larger than the number of players, envy-free allocations are likely to exist. If a fully envy-free or proportional allocation does exist in a particular instance, it may be preferable to choose that allocation before resorting to weaker properties or randomization. In this paper, we address the question of determining whether a $c$-EF or a $c$-Prop allocation exists, and if so, finding one.

Communication complexity was first studied by~\cite{Yao79:Complexity}.
The paper most relevant to our work is~\cite{Nisan06:Communication},
which shows that maximizing social welfare requires exponential
communication, even for two players with submodular
valuations. Furthermore, they show that for general valuations, any
constant factor approximation of the social welfare requires
exponential communication to compute. Although they do not mention
envy-freeness, proportionality, or fair division, some of their arguments can be adapted
to prove exponential lower bounds for some (but not all) of the cases that we study.

% However, their results cover only two of the many settings we
% consider. In particular, adapting their results yields lower bounds
% only for (1) 1-EF for two players with submodular valuations and (2)
% $c$-EF for two players with general valuations, for any $c > 0$. For
% completeness, we provide our own proofs of these results
% (Theorems~\ref{thm:EF-n=2-submod-lower} and
% \ref{thm:EF-n=2-general-lower}, respectively), using a slightly
% different approach that fits in well with the rest of our lower
% bounds. Looking at the table of our results
% (Table~\ref{tbl:results}), it is clear that these two theorems are
% only a small fraction of our results. We are not aware of any other
% results related to the communication complexity of discrete fair
% division prior to this paper. 

A recent and complementary line of work is presented in \cite{Branzei17:Communication}. They study the communication complexity of fair division with \emph{divisible} goods (also known as ``cake cutting"), where each resource can be divided into arbitrarily small pieces. Their paper complements ours with no overlap. Together, our papers give a comprehensive picture of the communication complexity of fair division in both the indivisible and divisible models.

The organization of the rest of the paper is as follows. Section~\ref{sec:model} formally presents the model. Section~\ref{sec:prop-n=2-submod-upper} presents our 1-Prop protocol for submodular valuations. In Section~\ref{sec:EF-n=2-submod-upper}, we discuss the PAS for 1-EF for submodular valuations. Section~\ref{sec:lower-bound-approach} describes our general lower bound approach, and proves a lemma that we will use to prove lower bounds later on in a standardized way. Section~\ref{sec:EF-n=2-submod-lower} uses that lemma to prove hardness for 1-EF for submodular valuations, which shows that the PAS from Section~\ref{sec:EF-n=2-submod-upper} is optimal. Section~\ref{sec:n>2} shows that the problem is always hard for more than two players, even for submodular valuations and even in the randomized setting. The rest of the paper is focused on resolving the two player case. Section~\ref{sec:subadd} presents the upper and lower bounds for subadditive valuations. Section~\ref{sec:general} considers general valuations, and also handles the randomized two player setting. Table~\ref{tbl:results} will be complete after this section. Finally, we consider the maximin share property (to be defined later) in Section~\ref{sec:maximin}.

\section{Model}\label{sec:model}

We formally introduce the discrete fair division model in Section~\ref{sec:model-fair-division}, and the communication complexity model in Section~\ref{sec:model-cc}.

\subsection{Fair division}\label{sec:model-fair-division}
Let $[k]$ denote the set $\{1...k\}$. Let $N = [n]$ be the set of players, and let $M$ be the set of items, where $|M| = m$. We assume throughout the paper that items are indivisible, meaning that an item cannot be split among multiple players. Player $i$'s value for each subset of $M$ is specified by a \emph{valuation} $v_i: 2^M \to \mathbb{R}_{\geq 0}$. We refer to a subset of $M$ as a \emph{bundle}.

We assume throughout the paper that valuations obey monotonicity (adding items to a bundle cannot decrease the value of the bundle) and normalization ($v_i(\emptyset) = 0$), and that $v_i(M) > 0$. We refer to set of the valuations constrained only by these three properties as ``general valuations".

There are many commonly studied subclasses of valuations, such as subadditive and submodular. A valuation $v$ is subadditive if for all bundles $S$ and $T$, $v(S \cup T) \leq v(S) + v(T)$. Submodular valuations represent ``diminishing returns": $v$ is submodular if $v(T \cup \{g\}) - v(T) \leq v(S \cup \{g\}) - v(S)$ whenever $S \subseteq T$. Every submodular valuation is subadditive, but not every subadditive valuation is submodular. Thus a problem that is hard for subadditive valuations may become tractable if valuations are restricted to be submodular. Similarly, problems that are hard for general valuations may be easier for subadditive valuations.

%These valuation classes have the following relationship:
%\[
%\text{(submodular valuations) $\subseteq$ (subadditive valuations) $\subseteq$ (general valuations)}
%\]

An \emph{allocation} is a partition of $M$ into $n$ disjoint subsets $(A_1, A_2...A_n)$, where $A_i$ is the bundle allocated to player $i$. The goal is to find a ``fair" allocation. The two most prominent fairness notions for indivisible items are envy-freeness and proportionality. Envy-freeness states that no player strictly prefers another player's bundle to her own, and proportionality states that every player receives at least $1/n$ of her value for the entire set of items. We can also define approximate versions of these properties:

\begin{definition}
An allocation $A = (A_1...A_n)$ is $c$-EF for some $c \in [0,1]$ if for all $i,j \in N$, 
\[
v_i(A_i) \geq c\cdot v_i(A_j)
\]
\end{definition}

\begin{definition}
An allocation $A = (A_1...A_n)$ is $c$-Prop for some $c \in [0,1]$ if for all $i \in N$,
\[
v_i(A_i) \geq c\cdot \cfrac{v_i(M)}{n}
\]
\end{definition}
\noindent Thus $1$-EF is standard envy-freeness, and $1$-Prop is standard proportionality.

We will say that a player is $(c,P)$-\emph{happy} with an allocation $A$ if property $c$-$P$  is satisfied from her viewpoint. Specifically, when $P = $ EF, we will say that player $i$ is $(c,P)$-happy with allocation $A$ if $v_i(A_i) \geq c\cdot v_i(A_j)$ for all $j$. For $P$ = Prop, we will say a player $i$ is $(c,P)$-happy if $v_i(A_i)~\geq~\mfrac{c}{n} v_i(M)$. We will typically leave $P$ implicit, and just say that player $i$ is $c$-happy. We sometimes also leave $c$ implicit, and just say that player $i$ is happy.

A instance of \textsc{Fair Division} consists of a set of players $N$, a set of items $M$, player valuations $(v_1...v_n)$, a fairness property $P\in \{\text{EF, Prop}\}$, and a constant $c \in [0,1]$. The goal is to find an allocation satisfying $c$-$P$, or show that none exists.

\subsubsection{Two players}
We use the following additional terminology when $n=2$. For a player $i$, we will use $\overline{i}$ to denote the other player. For an allocation $A = (A_1, A_2)$, let $\overline{A}$ be the allocation $(A_2,A_1)$. Also, when $n=2$, knowing player $i$'s bundle uniquely determines the overall allocation, since player $\overi$ simply has every item not in player $i$'s bundle. Therefore, with slight abuse of notation, we say that player $i$ is $c$-happy with bundle $S$ if player $i$ is $c$-happy with the allocation $A$ where $A_i = S$ and $A_{\overi} = M\backslash S$.

%%%%%%%%%%%%%%%%%%%%%%%%%%%%%%%%%%%%%
\subsection{Communication complexity}\label{sec:model-cc}

We assume that each player knows only her own valuation $v_i$, and does not know anything about other players' valuations. In order to solve an instance of \textsc{Fair Division}, players will need to exchange information about their valuations. We assume that all players know $N, M, P$, and $c$. Since there $2^m$ subsets of $M$, specifying a bundle requires $m$ bits. We will use $\valsize$ to refer to the number of bits required to represent a value $v_i(S)$. We assume that $\valsize$ is polynomial in $m$, otherwise sending even a single value would rule out a polynomial communication protocol.

A (deterministic) protocol $\Gamma$ specifies which player should speak (and what she should say) as a function of the messages sent so far, and terminates when a player declares that an allocation $A$ satisfies $c$-$P$, or when a player declares that no $c$-$P$ allocation exists. For fixed $N,M,P$, and $c$, we define the communication cost of a protocol $\Gamma$ to be the maximum number of bits $\Gamma$ sends across all player valuations $v_1...v_n$. Formally, let $C_{\Gamma}(N, M, (v_1...v_n), P, c)$ be the number of bits that $\Gamma$ communicates when run on the fair division instance $(N,M,(v_1...v_n),P,c)$. Then the communication cost of $\Gamma$ is $\max_{(v_1...v_n)} C_{\Gamma}(N, M, (v_1...v_n), P, c)$.

We define the deterministic communication complexity $D(n,m,P,c)$ as the minimum communication cost of any protocol $\Gamma$ which correctly solves \textsc{Fair Division} for $n$ players, $m$ items, property $P$ and constant $c$. Formally,
\[
D(n,m,P,c) = \min_{\Gamma} \max_{(v_1...v_n)} C_{\Gamma}([n], [m], (v_1...v_n), P, c)
\]
where $\Gamma$ ranges over all correct deterministic protocols.

In a randomized protocol $\Gamma_R$, each player also has access to an infinite stream of random bits. The protocol should correctly solve \textsc{Fair Division} with probability 2/3 (say) over these random bits. Like the deterministic setting, the communication cost of $\Gamma_R$ is the number of bits $\Gamma_R$ communicates for a worst-case choice of $v_1...v_n$. We can similarly define the randomized communication complexity $R(n,m,P,c)$ as the minimum communication cost of any randomized protocol $\Gamma_R$ which correctly solves \textsc{Fair Division} with probability at least 2/3. Formally,
\[
R(n,m,P,c) = \min_{\Gamma_R} \max_{(v_1...v_n)} C_{\Gamma_R}([n], [m], (v_1...v_n), P, c)
\]
where $\Gamma_R$ ranges over all correct randomized protocols. If valuations are restricted to be subadditive or submodular, the problem may become easier, so $D(n,m,P,c)$ and $R(n,m,P,c)$ may be affected. We use $D_{subadd}(n,m,P,c)$ and $D_{submod}(n,m,P,c)$ to denote the deterministic communication complexity when valuations are restricted to be subadditive and submodular, respectively (and similarly for $R_{subadd}(n,m,P,c)$ and $R_{submod}(n,m,P,c)$). The following relationships are immediate, for all $n,m,P,$ and $c$:
\begin{align*}
R(n,m,&P,c)  \leq D(n,m,P,c)\\
D_{submod}(n,m,P,c) \leq&\ D_{subadd}(n,m,P,c) \leq D(n,m,P,c)\\
R_{submod}(n,m,P,c) \leq&\ R_{subadd}(n,m,P,c) \leq R(n,m,P,c)
\end{align*}

Another factor that may affect the communication complexity is how the players gain access to random bits. In the \emph{public-coin} model, the players can also see other players' streams of random bits; in the \emph{private-coin} model, each player sees only her own stream. This distinction is not significant in our setting, however, due to the following theorem by \cite{Newman91:Private}.

\begin{theorem}[\cite{Newman91:Private}]
Suppose there exists a public-coin randomized protocol with communication cost $C$ on $\ell$ bits of input. Then there exists a private-coin randomized protocol with communication cost $O(C + \log \ell)$.
\end{theorem}
\noindent Thus we will assume all randomized protocols to be public-coin for the rest of the paper.

Finally, we mention the multiparty (i.e., $n > 2$) communication complexity model. There is more than one such model: for example, do players communicate in a peer-to-peer fashion, or is each message broadcast for all of the players to see? We discuss in Section~\ref{sec:n>2} how this turns out not to matter in our setting.

%Our lower bounds for $n>2$ will hold even when only two players have private valuations, and the valuations of all other players are assumed to be public. Thus we only ever use the two-party communication complexity model. This is discussed more in Section~\ref{sec:n>2}.

%As one might expect, the hardness of these problems is non-decreasing with $c$: a larger $c$ implies that we are asking for a stronger guarantee. Thus it is sufficient to characterize the maximum value of $c$ for which $D(n,m,P,c)$ and $R(n,m,P,c)$ are polynomial in $m$. Table~\ref{tbl:results} summarizes our results.

\section{An upper bound for 1-Prop with submodular valuations}\label{sec:prop-n=2-submod-upper}

This section presents our first result: a deterministic protocol for 1-Prop, when there are two players, and when valuations are submodular. The protocol will communicate just $m + 1$ values and a single bundle. Our protocol either finds a 1-Prop allocation or a $c^*$-Prop allocation. Recall that $c^*$ is the maximum $c$ such that a $c$-$P$ allocation exists.\footnote{Technically, our protocol always returns a $c^*$-Prop allocation, since we only consider $c \in [0,1]$. We state the 1-Prop case separately in the theorem because it is handled separately in the protocol.} We prove the following theorem:

% restatable environment so we can restate it later with the same number
\begin{restatable}{theorem}{propSubmodUpper}
\label{thm:prop-n=2-submod-upper}
For two players with submodular valuations, Protocol~\ref{pro:prop} has communication cost at most $(m + 1) v^{size} + m$, and either returns a 1-Prop allocation or a $c^*$-Prop allocation. This also implies that for any $c \in [0,1]$,
\[
D_{submod}(2,m,\text{Prop}, c) \le (m + 1) v^{size} + m
\]
\end{restatable}

To see that the theorem also implies $D_{submod}(2,m,\text{Prop}, c) \le (m + 1) v^{size} + m$ for any $c$, suppose that the protocol returns a $c^*$-Prop allocation where $c^* < 1$: then we know that no allocation is $c'$-Prop for any $c' > c^*$, so a $c$-Prop allocation exists if and only if $c^* \geq c$. Thus Protocol~\ref{pro:prop} either finds a $c$-Prop allocation or shows that none exists, for any $c \in [0,1]$.

 It will be important that the following condition is satisfied in this setting:
\begin{restatable}{condition}{condHappyWithOne}
\label{cond:happy-with-one}
For every allocation $A$, each player is happy with at least one of $A$ and $\overline{A}$. 
\end{restatable}
Recall that for an allocation $A = (A_1, A_2)$, $\overline{A} = (A_2, A_1)$. This condition is satisfied for proportionality with subadditive valuations (and hence also satisfied for submodular valuations):
\[
\max\big(v_i(A_1), v_i(A_2)\big) \geq \frac{1}{2}\big(v_i(A_1) + v_i(A_2)\big) \geq \frac{1}{2}v_i(A_1 \cup A_2) = \frac{1}{2}v_i(M) \geq \frac{c}{2}v_i(M)
\]
for all $c \in [0,1]$. Thus player $i$ is always happy if she receives the bundle $\argmax\big(v_i(A_1), v_i(A_2)\big)$. Also, we assume in this section that $v_1(M) = 1$, without loss of generality: were this not the case, we could simply rescale $v_1$ as needed.

%For two players, knowing player $i$'s bundle uniquely determines the overall allocation, since player $\overi$ simply has every item not in player $i$'s bundle. Therefore, with slight abuse of notation, we say that player $i$ is $c$-happy (or just ``happy") with bundle $S$ if player $i$ is $c$-happy with the allocation $A$ where $A_i = S$ and $A_{\overi} = M\backslash S$.

\begin{figure}
\centering
\begin{tabular}{ c|ccc} 
$k$ & 1 & 2 & 3\\
\hline
$v_1(g_1,\dots, g_k)$ & 2 & 2 & 3\\
%\hline
$\delta_k^M$ & 2 & 0 & 1
%\hline
\end{tabular}
\caption{An example of a possible valuation $v_1$ over three goods, and the corresponding values for $\delta_k^S$.}
\label{fig:delta}
\end{figure}

Let $M = (g_1, g_2...g_m)$ be an arbitrary ordering of the items. Consider starting from the empty set and adding the items in $M$ one at a time in this order. We define $\delta_k^M$ as player 1's marginal value of adding $g_k$ in this process: $\delta_k^M = v_1(g_1, g_2\dots g_{k-1}, g_k) - v_1(g_1, g_2\dots g_{k-1})$. Note that $\delta_k^M$ is not equal to $v_i(\{g_k\})$ in general, because of submodularity.

%For example, $\delta_1^M = v_1(\{g_1\})$ and $\delta_m^M = v_1(M) - v_1(M\backslash\{g_m\})$.

The protocol is as follows.  The first step is common to all of our deterministic protocols: player 1 checks if there is an allocation $A$ where she is happy with both $A$ and $\overline{A}$. If so, player 2 can choose whichever she prefers, and we are done by Condition~\ref{cond:happy-with-one}. If this fails, the following condition is satisfied:
\begin{restatable}{condition}{cutFailedCond}
\label{cut-failed-cond}
There is no allocation $A$ for which player 1 is happy with both $A$ and $\overline{A}$.
\end{restatable}

In this case, player 1 sends the values $(\delta_1^M, \delta_2^M...\delta_m^M)$ to player 2. For every bundle $S$, player 2 needs to be able to figure out whether player 1 likes $S$ or $\mss$. To do this, player 2 simply pretends that player 1's valuation is additive where $\delta_k^M$ is the value of item $g_k$. Formally, let $\chi(S) = \sum\limits_{g_k \in S} \delta_k^M$: player 2 pretends that $v_1(S) = \chi(S)$. This will not be a perfect estimate of $v_1$, of course, but player 2 does not need to know the exact value of $v_1(S)$: she only needs to know whether player 1 is happy with $S$. 

Lemma~\ref{lem:prop-chi} shows that this actually works: assuming Condition~\ref{cut-failed-cond}, $v_1(S) \geq 1/2$ if and only if $\chi(S) \geq 1/2$. We informally argue why this is case. Crucially, submodularity implies that $\chi(S)$ will be an underestimate of $v_1(S)$: $v_1(S) \geq \chi(S)$ for all $S$. Since $\chi(S) + \chi(\mss) = \sum_{k = 1}^m \delta_k^M = v_1(M)$, either $\chi(S) \geq 1/2$ or $\chi(\mss) \geq 1/2$. Say $\chi(S) \geq 1/2$: then $v_1(S) \geq \chi(S) \geq 1/2$, so player 1 is happy with $S$. Then by Condition~\ref{cut-failed-cond}, we know that player 1 is not happy with $\mss$. Therefore, for any bundle $S$, player 2 can correctly use $\chi$ as a proxy for $v_1$ to determine which of $S$ and $\mss$ player 1 is happy with. Thus $\chi$ is sufficient for her to determine whether or not a 1-Prop allocation exists, and if so, find one. This lemma is the heart of Protocol~\ref{pro:prop}.

Step 4, $S^*(v_i)$, and $\bfc_i(S^*(v_i))$ are necessary only for finding a $c^*$-Prop allocation if no 1-Prop allocation is found. For a bundle $S$ and property $P$, let $\bfc_i^P(S)$ be the maximum $c' \leq 1$ such that player $i$ is $c'$-happy with $S$. For example, $\bfc_i^{Prop}(S) = \min\left(1, \mfrac{2v_i(S)}{v_i(M)}\right) = \min(1, 2v_i(S))$, since we assumed $v_i(M) = 1$. Although this section considers only proportionality, we allow for either $P \in$ \properties in our definitions, since we will use this terminology again in later sections. We will typically leave $P$ implicit and write $\bfc_i(S)$.

For each player $i$, we define a special bundle
\[
S^*(v_i) = \argmax_{S \subseteq M:\ \bfc_i(S) < c} \bfc_i(S)
\]
In words, $S^*(v_i)$ the bundle that player $i$ is the most happy with, out of all of the bundles she is not fully happy (i.e., $c$-happy) with. 

%%%%%%%%%%%%%%%%%%%%%%%%%%%%%%%%%%%%%%%%
%%%%%%%%%%%%%%%%%%%%%%%%%%%%%%%%%
% THE PROTOCOL
\begin{protocol}
\protocolspacing
\begin{center}
Private inputs: $v_1, v_2$\\
Public inputs: the ordering of $M = \{g_1, g_2...g_m\}$
\end{center}

\begin{enumerate}
\item If there exists an allocation $A$ where player 1 is happy with both $A$ and $\overline{A}$, player 1 sends that allocation to player 2. If player 2 is happy with $A$, she declares that $A$ is 1-Prop, otherwise she declares that $\overline{A}$ is 1-Prop.

\item If there is no such allocation $A$, player 1 sends the values $(\delta_1^M, \delta_2^M...\delta_m^M)$ to player 2, along with $S^*(v_1)$ and the value $\bfc_1(S^*(v_1))$.

\item Player 2 first checks if there exists any bundle $S$ where $\chi(S) \geq 1/2$ and $v_2(\mss) \geq 1/2$. If so, she declares that the allocation $(S, M\backslash S)$ is 1-Prop.

\item If not, player 2 computes $S^*(v_2)$, $\bfc_2(S^*(v_2))$, and $i =\argmax_{i' \in \{1,2\}} \bfc_{i'}(S^*(v_{i'}))$. Let $A$ be the allocation where $A_{i} = S^*(v_{i})$ and $A_{\overline{i}} = M\backslash S^*(v_{i})$. Player 2 then declares that $A$ is $\bfc_{i}(S^*(v_{i}))$-Prop, and that $c^* = \bfc_{i}(S^*(v_{i}))$.

\end{enumerate}
\caption{Protocol for two players with submodular valuations to either find a 1-Prop allocation or a $c^*$-Prop allocation.}
\label{pro:prop}
\end{protocol}

%%%%%%%%%%%%%%%%%%%%%%%%%%%%%%%%%%%%%%%%
It will be useful for the analysis to define $\delta_i^S$ for an arbitrary bundle $S$. First, let
\[S_{\le k} = \{g_j \in S\ |\ j \leq k\}
\]
For example, $S_{\le 0} = \emptyset$ and $S_{\le m} = S$ for all $S$. Also, whenever $g_k \in S$, we have $S_{\le k} = S_{\le k-1} \cup\{g_k\}$. Let $\delta_k^S = v_1(S_{\le k}) - v_1(S_{\le k-1})$. Note that for all $S$, $v_1(S) = \sum_{k = 1}^m \delta_k^S$.

\begin{lemma}\label{lem:prop-chi}
Assuming Condition~\ref{cut-failed-cond}, for any bundle $S$, $v_1(S) \geq 1/2$ if and only if $\chi(S) \geq 1/2$.
\end{lemma}

\begin{proof}
We first claim that for any bundle $S$ and any item $g_k \in S$, $\delta_k^S \ge \delta_k^M$. We have
\[
\delta_k^S = v_1(S_{\le k}) - v_1(S_{\le k-1}) = v_1(S_{\le k-1}\cup \{g_k\}) - v_1(S_{\le k-1})
\]
and
\[
\delta_k^M = v_1(M_{\le k}) - v_1(M_{\le k-1}) = v_1(M_{\le k-1}\cup \{g_k\}) - v_1(M_{\le k-1})
\]
Since $S_{\leq k-1} \subseteq M_{\leq k-1}$, we have $v_1(S_{\le k-1}\cup \{g_k\}) - v_1(S_{\le k-1}) \geq v_1(M_{\le k-1}\cup \{g_k\}) - v_1(M_{\le k-1})$ by submodularity. Thus $\delta_k^S \geq \delta_k^M$ for all $k$ and $S$. Therefore for any bundle $S$,
\[
v_1(S) = \sum\limits_{g_k \in S} \delta_k^S \geq \sum\limits_{g_k \in S} \delta_k^M 
= \chi(S)
\]
so $v_1(S) \geq \chi(S)$ for all $S \subseteq M$.

Suppose $\chi(S) \geq 1/2$: then we immediately have $v_1(S) \geq 1/2$ by the above argument. Suppose $v_1(S) \geq 1/2$. Then by Condition~\ref{cut-failed-cond}, $v_1(\mss) < 1/2$. Therefore $\chi(\mss) < 1/2$. Next, we have
\[
\chi(S) + \chi(\mss) = \sum\limits_{g_k \in S} \delta_k^M + \sum\limits_{g_k \in \mss} \delta_k^M = \sum_{k = 1}^m \delta_k^M = v_1(M) = 1
\]
Since $\chi(\mss) < 1/2$, we have $\chi(S) \geq 1/2$.
\end{proof}

%%%%%%%%%%%%%%%%%%%%%%%%%%%%%%%%%
%%%%%%%%%%%%%%%%%%%%%%%%%%%%%%%%%
%%%%%%%%%%%%%%%%%%%%%%%%%%%%%%%%%
% FINAL THEOREM

\propSubmodUpper*

\begin{proof}
If the protocol terminates in step 1, just one bundle is communicated (and zero values), which requires $m$ bits. Thus in this case, the communication cost is $m \leq (m+1)v^{size} + m$. If the protocol does not terminate in step 1, then the $m$ values $(\delta_1^M...\delta_m^M)$ are sent, plus the bundle $S^*(v_1)$, plus the value $\bfc_1(S^*(v_1))$. By definition of $\bfc_1^{Prop}$, $\bfc_1(S^*(v_1))$ requires a single value to communicate. 

Thus in this case, $m+1$ values and one bundle are communicated, so the communication cost is $(m+1)v^{size} + m$. Therefore the communication cost bound is satisfied.

It remains to prove correctness. Suppose the protocol terminates in step 1. By Condition~\ref{cond:happy-with-one}, player 2 is happy with at least one of $A$ and $\overline{A}$. Therefore player 2 is happy with whichever of $A$ and $\overline{A}$ she declares to be 1-Prop. Player 1 is happy with both $A$ and $\overline{A}$, so she is also happy. Therefore if the protocol terminates in step 1, the declared allocation is in fact 1-Prop.

Suppose the protocol does not terminate in step 1. We assume Condition~\ref{cut-failed-cond} for the remainder of the proof. Suppose player 2 declares that $(S, \mss)$ is 1-Prop in step 3: then
\[
\chi(S) \geq 1/2\ \ \textnormal{and}\ \ v_2(\mss) \geq 1/2
\]
Thus by Lemma~\ref{lem:prop-chi}, $v_1(S) \geq 1/2$, so $(S, \mss)$ is indeed a 1-Prop allocation.

So suppose the protocol does not terminate until step 4. We first claim that no 1-Prop allocation exists. Suppose that a 1-Prop allocation $A$ does exist: then $v_i(A_i) \geq 1/2$ for both $i$. Since the protocol did not terminate in step 1, we have Condition~\ref{cut-failed-cond}. Thus by Lemma~\ref{lem:prop-chi}, $\chi(A_1) \geq 1/2$. Let $S = A_1$: then
\[
\chi(S) \geq 1/2\ \ \textnormal{and}\ \ v_2(\mss) = v_2(A_2) \geq 1/2
\]
so the protocol should have terminated in step 3, which is a contradiction.

Therefore no 1-Prop allocation exists. It remains to show that we return a $c^*$-Prop allocation in this case. Let $i =\argmax_{i' \in \{1,2\}} \bfc_{i'}(S^*(v_{i'}))$ as computed by player 2 in step 4. Let $A$ be the allocation returned by the protocol in this case: $A_{i} = S^*(v_{i})$ and $A_{\overline{i}} = M\backslash S^*(v_{i})$.

We first claim that $A$ is $\bfc_i(S^*(v_i))$-Prop. Player $i$ is $\bfc_i(S^*(v_i))$-happy with $A$ by definition, and we claim that player $\overi$ is 1-happy with $A$. If $\overi$ were not 1-happy with $A$, then she must be 1-happy with $\overline{A}$ by Condition~\ref{cond:happy-with-one}. Furthermore, Since player $i$ is not 1-happy with $A$, she must be 1-happy with $\overline{A}$ also by Condition~\ref{cond:happy-with-one}. But then both players are 1-happy with $\overline{A}$, which is a contradiction.

Thus $A$ is $\bfc_i(S^*(v_i))$-Prop. Suppose that $c^* \neq \bfc_i(S^*(v_i))$: then there exists an allocation $A'$ where $A'$ is $c$-Prop for some $c > \bfc_i(S^*(v_i))$. We know that player $i$ cannot be happier than $\bfc_i(S^*(v_i))$-happy without being 1-happy, so player $i$ must be 1-happy with $A'$. That implies that player 2 is not 1-happy with $A'$, since no allocation makes both players 1-happy in this case. But then the happiest player $\overi$ can be is $\bfc_{\overi}(S^*(v_{\overi}))$, and $\bfc_{\overi}(S^*(v_{\overi}))\leq \bfc_i(S^*(v_i))$ by assumption. Thus for any allocation, there is a player who is at most $\bfc_i(S^*(v_i))$-happy. Therefore no allocation is $c$-Prop for any $c > \bfc_i(S^*(v_i))$.
\end{proof}

\section{PAS for EF with submodular valuations}\label{sec:EF-n=2-submod-upper}

In this section, we prove our other positive result for specifically submodular valuations: a deterministic protocol for $c$-EF when $c < 1$, and when there are two players. This is our most technically involved result. We prove the following theorem:

% this uses the "restatable" package to restate the theorem without adding a
% new number
\begin{restatable}{thm}{EFsubmodUpper}
\label{thm:EF-n=2-submod-upper}
For two players with submodular valuations and any $c < 1$, Protocol~\ref{pro:minimal-upper} has communication cost at most $2m(m+1)^{\frac{8}{1-c}} + 2\valsize$, and either returns a $c$-EF allocation or a $c^*$-EF allocation. This also implies that
\[
D_{submod}(2,m,\text{EF}, c) \le 2m(m+1)^{\frac{8}{1-c}} + 2\valsize
\]
for any $c < 1$.
\end{restatable}

This constitutes a polynomial-communication approximation scheme (PAS): the communication cost approaches infinity exponentially as $c$ goes to $1$, but for any fixed constant $c < 1$, it is polynomial in $m$.\footnote{Because the dependence on $\frac{1}{1-c}$ is exponential, this constitutes a PAS but not an FPAS. An FPAS is ruled out in Section~\ref{sec:EF-n=2-submod-lower}.}

We use much of the same terminology from Section~\ref{sec:prop-n=2-submod-upper}: in particular, $\bfc_i^P(A)$, $S^*(v_i)$, $S_{\le k}$, and $\delta_k^S$. Also, recall the following condition:

\condHappyWithOne*

This is satisfied for $c$-EF for any $c \in [0,1]$, even for general valuations: if $v_i(A_i) \geq v_i(A_{\overi})$, player $i$ is happy with $A$. Otherwise, $v_i(A_{\overi}) \geq v_i(A_i)$, so she is happy with $\overline{A}$.

Our PAS protocol will use the minimal bundle analysis discussed in Section~\ref{sec:techniques}. For a fixed constant $c$, we say that a bundle $S$ is \emph{minimal} for a particular player if that player is $c$-happy with $S$, but for all $g \in S$, she is not $c$-happy with $S\backslash \{g\}$. We use $\cals$ to denote the set of player 1's minimal bundles: each $S \in \cals$ is a minimal bundle for player 1. Also, in this section, we assume that $v_1(M) = 1$.

\begin{figure}
\centering
\begin{tabular}{ c|ccc} 
bundle $S$ & $v_i(S)$ & 1-Prop? & minimal? \\
\hline
$\{g_1\}$ & 2 & no & N/A\\
$\{g_2\}$ & 4 & no & N/A\\
$\{g_3\}$ & 5 & yes & yes\\
$\{g_1, g_2\}$ & 6 & yes & yes\\
$\{g_1, g_3\}$ & 7 & yes & no\\
$\{g_2, g_3\}$ & 9 & yes & no\\
$\{g_1, g_2, g_3\}$ & 10 & yes & no
\end{tabular}
\caption{An example demonstrating the minimal bundle property for $P =$ Prop and $c=1$. This instance involves a valuation $v_i$ over three goods. Since $v_i(M) = v_i(\{g_1, g_2, g_3\}) = 10$ in this case, player $i$ is happy with $S$ if and only if $v_i(S) \ge 10/n = 5$. For example, player $i$ is happy with $\{g_1, g_3\}$, but that bundle is not minimal, since player $i$ is also happy with $\{g_3\}$. In contrast, $\{g_1, g_2\}$ is minimal, since player $i$ is happy with neither $\{g_1\}$ nor $\{g_2\}$.}
\label{fig:minimal}
\end{figure}

%%%%%%%%%%%%%%%%%%%%%%%%%%%%%%%%%%%%%
\subsection{The protocol}

We now describe Protocol~\ref{pro:minimal-upper}, also known as the Minimal Bundle Protocol. Although we only consider envy-freeness in this section, we define Protocol~\ref{pro:minimal-upper} for either $P \in$ \properties. We will use this same protocol in Section~\ref{sec:n=2-subadd-upper} to prove upper bounds for both envy-freeness and proportionality in the subadditive case.

First, if there is an allocation $A$ where player 1 is happy with both $A$ and $\overline{A}$, we are done: player 2 chooses her favorite of $A$ and $\overline{A}$, and she is guaranteed to be happy with at least one them by Condition~\ref{cond:happy-with-one}. If there is no such allocation $A$, player 1 sends the set $\cals$ of all of her minimal bundles to the other player. We will prove that in our setting, the number of minimal bundles sent in step 2 must be polynomial in $m$. Specifically, we will show that  $|\cals | < 2(m+1)^{\frac{8}{1-c}}$.

The minimal bundles represent the most player 1 is willing to compromise while still being happy: she does not require anything more than a minimal bundle, but she is not happy with any strict subset of any of her minimal bundles. In this way, receiving a minimal bundle is both necessary and sufficient for player 1 to be happy. Using this reasoning, we will show that knowing $\cals$ is sufficient for player 2 to find a $c$-$P$ allocation or show that none exists. Finally, step 4 is identical to that of Protocol~\ref{pro:prop}, and is used to find a $c^*$-$P$ allocation when no $c$-$P$ allocation exists.

%%%%%%%%%%%%%%%%%%%%%%%%%%%%%%%%%%%%
%%%%%%%%%%%%%%%%%%%%%%%%%%%%%%%%%%%%%
% THE PROTOCOL
\newcounter{pro:minimal-upper}
\setcounter{pro:minimal-upper}{\value{algorithm}}
% the purpose of these counters is so I can restate this later in a kind of hacky
% way because I couldn't find a better way to do it
\begin{restatable}[tb]{protocol}{proMinimalUpper}
\begin{center}
\protocolspacing
Private inputs: $v_1, v_2$\\
Public inputs: $P, c$
\end{center}

\begin{enumerate}
\item If there exists an allocation $A$ where player 1 is happy with both $A$ and $\overline{A}$, player 1 sends that allocation to player 2. If player 2 is happy with $A$, she declares that $A$ is $c$-$P$, otherwise she declares that $\overline{A}$ is $c$-$P$.

\item If there is no such allocation $A$, player 1 sends the set $\cals$ of her minimal bundles to player 2. She also sends the bundle $S^*(v_1)$ and the value $\bfc_{1}(S^*(v_{1}))$.

\item Player 2 first checks if there exists a bundle $S \in \cals$ where player 2 is happy with $M\backslash S$. If so, she declares that $(S ,M\backslash S)$ is $c$-$P$.

\item If not, player 2 computes $S^*(v_2)$ and $i =\argmax_{i' \in \{1,2\}} \bfc_{i'}(S^*(v_{i'}))$. Let $A$ be the allocation where $A_{i} = S^*(v_{i})$ and $A_{\overline{i}} = M\backslash S^*(v_{i})$. Player 2 then declares that $A$ is $\bfc_{i}(S^*(v_{i}))$-$P$, and that $c^* = \bfc_{i}(S^*(v_{i}))$.

\end{enumerate}
\caption{Protocol for two players to either find a $c$-$P$ allocation or a $c^*$-$P$ allocation.}
\label{pro:minimal-upper}
\end{restatable}

%%%%%%%%%%%%%%%%%%%%%%%%%%%
%%%%%%%%%%%%%%%%%%%%%%%%%%%
% CORRECTNESS

\subsection{Correctness}

We now formally prove the correctness of Protocol~\ref{pro:minimal-upper}. We will prove a few helpful lemmas before proving the main correctness lemma (Lemma~\ref{lem:minimal-correct}).

% I turned out not to need these to be restatable
\begin{restatable}{lemma}{lemMinimalSound}
\label{lem:minimal-sound}
If Protocol~\ref{pro:minimal-upper} declares an allocation to be $c$-$P$, the allocation is in fact $c$-$P$.
\end{restatable}

\begin{proof}
The only two steps that can declare an allocation to be $c$-$P$ are steps 1 and 3. Suppose the protocol declares an allocation to be $c$-$P$ in step 1.  Then by assumption, there exists an allocation $A$ where player 1 is happy with both $A$ and $\overline{A}$. If player 2 declares $A$ to be $c$-$P$, then both players are happy with $A$, and the claim is satisfied. If player 2 declares $\overline{A}$ to be $c$-$P$, then she was not happy with $A$. By Condition~\ref{cond:happy-with-one}, player 2 is happy with $\overline{A}$. Thus $\overline{A}$ is $c$-$P$ in this case, so the lemma is satisfied if the protocol terminates in step 1.

Suppose the protocol declares an allocation to be $c$-$P$ in step 3. Then the allocation declared can be written as $(S, M\backslash S)$ for some $S \in \cals$. Since $S$ is minimal, player 1 is happy with $S$ by assumption, and player 2 only declares an allocation to be $c$-$P$ in this step if she is happy with it. Thus the lemma is satisfied in this case as well.
\end{proof}

%%%%%%%%%%%%%%%%%%%%%%%%%%%%%%%%%%%%%
\begin{restatable}{lemma}{lemMinimalSubset}
\label{lem:minimal-subset}
Player 1 is happy with a bundle $S$ if and only if there exists a minimal bundle $T$ where $T \subseteq S$.
\end{restatable}

\begin{proof}
$(\implies)$ Suppose player 1 is happy with bundle $S$. If $S$ is minimal, we are done, so assume $S$ is not minimal. Then there exists $g\in S$ where player 1 is happy with $S\backslash\{g\}$. If $S\backslash\{g\}$ is not minimal, there again exists some $g' \in S\backslash \{g\}$ that we can remove, and this process can be repeated until we obtain some minimal subset of $S$.

$(\impliedby)$ Suppose there exists a minimal bundle $T$ where $T \subseteq S$. Then by monotonicity, $v_1(S) \geq v_1(T)$. Since $T$ is minimal, player 1 is happy with $T$. If $P =$ Prop, this is sufficient to show that player 1 is happy with $S$. If $P =$ EF, it is also necessary to note that $v_1(M\backslash S) \leq v_1(M\backslash T)$, again by monotonicity. Thus the claim holds for both $P \in$ \properties.
\end{proof}

%%%%%%%%%%%%%%%%%%%%%%%%%%%%%%%%%%%%%

\begin{restatable}{lemma}{lemMinimalComplete}
\label{lem:minimal-complete}
Protocol~\ref{pro:minimal-upper} declares an allocation to be $c$-$P$ if and only if a $c$-$P$ allocation exists.
\end{restatable}

\begin{proof}
If no $c$-$P$ allocation exists, the protocol does not declare any allocation to be $c$-$P$ by Lemma~\ref{lem:minimal-sound}. Thus assume a $c$-$P$ allocation $A$ exists. Then player 1 is happy with $A_1$, so by Lemma~\ref{lem:minimal-subset}, there exists $S \in \cals$ where $S \subseteq A_1$. Then $A_2 \subseteq M\backslash S$, so by monotonicity, player 2 is happy with $M\backslash S$. Thus if the protocol has not already terminated, player 2 will declare will declare $(S, M\backslash S)$ to be $c$-$P$. Then by Lemma~\ref{lem:minimal-sound}, the declared allocation is in fact $c$-$P$, so the claim is satisfied in this case.

If the protocol terminated before player 2 considered $S$ in step 3, the protocol declared some other allocation to be $c$-$P$, and the declared allocation is again $c$-$P$ by Lemma~\ref{lem:minimal-sound} in this case. Thus the claim is satisfied in both cases.
\end{proof}

%%%%%%%%%%%%%%%%%%%%%%%%%%%%%%%%%%%%
Finally, we show that the protocol correctly returns a $c^*$-$P$ allocation if no $c$-$P$ allocation exists. Recall the definitions of $S^*(v_i)$ and $\bfc(S)$: $\bfc_i(S)$ is the maximum $c' \leq 1$ where player $i$ is $c'$-happy with $S$, and $S^*(v_i) = \argmax_{S\subseteq M:\ \bfc_i(S) < c} \bfc_i(S)$. In words, $S^*(v_i)$ is the bundle that makes player $i$ the most happy, out of all the bundles that do not make her $c$-happy. For $P =$ EF, $\bfc_i(S) = \min\left(1, \mfrac{v_i(S)}{v_i(\mss)}\right)$.

\begin{restatable}{lemma}{lemMinimalCorrect}
\label{lem:minimal-correct}
Protocol~\ref{pro:minimal-upper} either returns a $c$-$P$ allocation or a $c^*$-$P$ allocation.
\end{restatable}

\begin{proof}
If a $c$-$P$ allocation exists, Lemma~\ref{lem:minimal-complete} implies that the protocol correctly returns one, so the claim is satisfied in this case. 

Suppose no $c$-$P$ allocation exists: then the protocol does not declare an allocation to be $c$-$P$, again by Lemma~\ref{lem:minimal-complete}. Thus the protocol does not terminate until step 4. Let $i =\argmax_{i' \in \{1,2\}} \bfc_{i'}(S^*(v_{i'}))$ as computed by player 2 in step 4. Let $A$ be the allocation returned by the protocol in this case: $A_{i} = S^*(v_{i})$ and $A_{\overline{i}} = M\backslash S^*(v_{i})$.

First observe that $A$ is $\bfc_i(S^*(v_i))$-$P$: this is because player $i$ is $\bfc_i(S^*(v_i))$-happy with $A$, and player $\overi$ is $c$-happy with $A$. Suppose that $A$ is not $c^*$-$P$: then there exists an allocation $A'$ where $A'$ is $c''$-$P$ for some $c'' > \bfc_i(S^*(v_i))$. We know that player $i$ cannot be happier than $\bfc_i(S^*(v_i))$-happy without being $c$-happy, so player $i$ must be $c$-happy with $A'$. That implies that player 2 is not $c$-happy with $A'$, since no allocation makes both players $c$-happy in this case. But then the happiest player $\overi$ can be is $\bfc_{\overi}(S^*(v_{\overi}))$, and $\bfc_{\overi}(S^*(v_{\overi}))\leq \bfc_i(S^*(v_i))$ by assumption. Thus for any allocation, there is a player who is at most $\bfc_i(S^*(v_i))$-happy. Therefore $c^* = \bfc_i(S^*(v_i))$.
\end{proof}

%%%%%%%%%%%%%%%%%%%%%%%%%%%%%%%%%%%%%%
%%%%%%%%%%%%%%%%%%%%%%%%%%%%%%%%%%%%%%
%%%%%%%%%%%%%%%%%%%%%%%%%%%%%%%%%%%%%%
% COMMUNICATION BOUND
\subsection{Communication cost}\label{sec:pro-communication}

It remains to bound the communication cost. This will primarily consist of proving an upper bound on the number of minimal bundles player 1 sends to player 2. We will go through a series of helpful lemmas before proving the final theorem.

The upper bound on the number of minimal bundles will depend on there being no allocation $A$ for which player 1 is happy with both $A$ and $\overline{A}$: recall that if there is such an allocation, then Protocol~\ref{pro:minimal-upper} terminates after step 1 and does not even send the set of minimal bundles $\cals$. This condition was defined in Section~\ref{sec:prop-n=2-submod-upper}.

\cutFailedCond*

Let $\Delta(S, g)$ be player 1's marginal value for adding item $g$ to bundle $S$. Formally, $\Delta(S, g) = v_1(S \cup \{g\}) - v_1(S)$. Also, let $\alpha = \mfrac{1-c}{2}$.

The idea behind Lemma~\ref{lem:marginal-value} is the following. Because of Condition~\ref{cut-failed-cond}, we have $c \cdot v_1(S) > v_1(\mss)$ whenever player 1 is happy with $S$. If $S$ is minimal, then moving any $g \in S$ to $\mss$ will invert this inequality: $v_1(S\backslash\{g\}) < c\cdot v_1\big((\mss) \cup\{g\}\big)$. Lemma~\ref{lem:marginal-value} uses this to show that at least one of $\Delta(S\backslash \{g\}, g)$ and $\Delta(\mss, g)$ has to be fairly large.

%%%%%%%%%%%%%%%%%%%%%%%%%%%%%%%%%%%%%%%%
\begin{restatable}{lemma}{lemMarginalValue}
\label{lem:marginal-value}
Assuming Condition~\ref{cut-failed-cond}, for every minimal bundle $S$ and every good $g \in S$,
\[
\max\Big(\Delta(S\backslash \{g\}, g), \Delta(\mss, g)\Big) \geq \alpha
\] 
\end{restatable}

\begin{proof}
Since $S$ is minimal, for every good $g \in S$, we know that player 1 is not happy with $S\backslash \{g\}$. Specifically,
\[
v_1(S\backslash \{g\}) < c\cdot v_1\big((\mss)\cup\{g\}\big)
\]
so by definition of $\Delta$, we have
\[
v_1(S) - \Delta(S\backslash\{g\}, g) < c\cdot \big(v_1(\mss) + \Delta(\mss, g)\big)
= c\cdot v_1(\mss) + c\cdot \Delta(\mss, g)
\]

We also know that player 1 is happy with $S$. Thus by Condition~\ref{cut-failed-cond}, player 1 is not happy with $\mss$, so $v_1(\mss) < c \cdot v_1(S)$. Adding this to the above equation yields
\begin{align*}
v_1(S) - \Delta(S\backslash\{g\}, g) + v_1(\mss) 
	<&\ c\cdot v_1(\mss) + c\cdot \Delta(\mss, g) + c\cdot v_1(S)\\
(1-c)v_1(S) + (1-c)v_1(\mss) <&\  \Delta(S\backslash\{g\}, g) + c\cdot \Delta(\mss, g)\\
(1-c)v_1(S) + (1-c)v_1(\mss) <&\  \Delta(S\backslash\{g\}, g) +  \Delta(\mss, g)\\
(1-c)\big(v_1(S) + v_1(\mss)\big) <&\  \Delta(S\backslash\{g\}, g) + \Delta(\mss, g)\\
(1-c)v_1(M) <&\  \Delta(S\backslash\{g\}, g) +  \Delta(\mss, g)\\
\end{align*}
where the last step follows from submodularity (actually just subadditivity).

Since $v_1(M) = 1$ by assumption, we have
\begin{align*}
\Delta(S\backslash\{g\}, g) +  \Delta(\mss, g) \geq&\ 1-c\\
\max\Big(\Delta(S\backslash \{g\}), \Delta(\mss, g)\Big) \geq&\ \frac{1-c}{2} = \alpha
\end{align*}
\end{proof}

%%%%%%%%%%%%%%%%%%%%%%%%%%%%%%%%%%%%%%%%
Next, we define the a directed graph $G = (V, E)$ which we call the \emph{minimal bundle graph}. The vertex set $V$ is the set of minimal bundles. With slight abuse of notation, we will use $S$ and $T$ to refer both to minimal bundles and to the corresponding vertices in $V$. We define the edge set $E$ by
\[
E = \big\{(S,T)\ |\ \exists g\in S\ \textnormal{where}\ T\subseteq (\mss)\cup \{g\} \big\}
\]

The next three lemmas establish some useful properties of the minimal bundle graph.

\begin{restatable}{lemma}{lemGinT}
\label{lem:g-in-T}
Assuming Condition~\ref{cut-failed-cond}, let $(S,T) \in E$, and let $g$ be a good in $S$ such that $T\subseteq (\mss)\cup\{g\}$. Then $g \in T$.
\end{restatable}

\begin{proof}
Suppose $g\not\in T$: then $S\subseteq M\backslash T$. Since $S$ is minimal, player 1 is happy with $S$. Thus by monotonicity, player 1 is also happy with $M\backslash T$. But player 1 is also happy with $T$, because $T$ is minimal. This contradicts Condition~\ref{cut-failed-cond}, so we must have $g\in T$.
\end{proof}

%%%%%%%%%%%%%%%%%%%%%%%%%%%%%%%%%%%%%%%%
\begin{restatable}{lemma}{lemGUnique}
\label{lem:g-unique}
Assuming Condition~\ref{cut-failed-cond}, if $(S,T) \in E$, then there is a unique $g \in S$ where $T\subseteq (\mss)\cup\{g\}$.
\end{restatable}

\begin{proof}
Suppose there exist $g_1, g_2\in S$ where $g_1 \neq g_2$, $T\subseteq (\mss)\cup\{g_1\}$, and $T\subseteq (\mss)\cup\{g_2\}$. Then by Lemma~\ref{lem:g-in-T}, $g_1 \in T$ and $g_2 \in T$. But this contradicts $T\subseteq (\mss)\cup \{g_1\}$, because $g_2 \in S\backslash\{g_1\}$, so $g_2 \not \in (\mss)\cup \{g_1\}$. Therefore $g_1 = g_2$.
\end{proof}

Using Lemma~\ref{lem:g-unique} for each edge $(S,T) \in E$, let $g(S,T)$ be the unique good such that $T\subseteq (\mss)\cup\{g(S,T)\}$.

%%%%%%%%%%%%%%%%%%%%%%%%%%%%%%%%%%%%%%%%
\begin{restatable}{lemma}{lemGraphSymmetric}
\label{lem:graph-symmetric}
Assuming Condition~\ref{cut-failed-cond}, if $(S,T) \in E$, then $(T,S) \in E$. Furthermore, $g(T,S) = g(S,T)$.
\end{restatable}

\begin{proof}
Suppose $(S,T) \in E$: then $T\subseteq (\mss)\cup\{g(S,T)\}$. By Lemma~\ref{lem:g-in-T}, we have $g(S,T) \in T$. Since $T\subseteq (\mss)\cup \{g(S,T)\}$, we have $S\backslash\{g(S,T)\} \subseteq M\backslash T$. Therefore $S \subseteq (M\backslash T)\cup\{g(S,T)\}$, and so $(T, S) \in E$ and $g(S,T) = g(T,S)$.
\end{proof}

%%%%%%%%%%%%%%%%%%%%%%%%%%%%%%%%%%%%%%%%
The next lemma is because there are $|S|$ items in $S$ that we could move to $\mss$. The proof uses Lemma~\ref{lem:g-unique} to show that each of them will yield a different minimal bundle $T$, so this constitutes $|S|$ distinct edges $(S,T)$.

\begin{restatable}{lemma}{lemGraphOutdegree}
\label{lem:graph-outdegree}
The out-degree of each bundle $S \in V$ is at least $|S|$.
\end{restatable}

\begin{proof}
Let $S = \{g_1, g_2...g_{|S|}\}$. We first claim that for all $g_j \in S$, there exists $T_j \in V$ where $T_j \subseteq (\mss)\cup\{g_j\}$. Consider some $g_j \in S$. Because $S$ is minimal, we know that player 1 is not happy with $S \backslash \{g_j\}$. Therefore player 1 must be happy with $(\mss)\cup \{g_j\}$. Then by Lemma~\ref{lem:minimal-subset}, there exists $T_j \subseteq (\mss)\cup\{g_j\}$ where $T_j$ is minimal. Therefore $(S, T_j) \in E$.

By Lemma~\ref{lem:g-unique}, $g_j = g(S,T_j)$ is unique. Thus for all $g \in S$ where $g \neq g(S,T_j)$, we have $g \not \in T_j$. This implies that each $T_j$ is distinct. Thus $(S, T_1), (S,T_2)...(S, T_{|S|})$ are all distinct edges in $E$, so the out-degree of $S$ is at least $|S|$.
\end{proof}

%%%%%%%%%%%%%%%%%%%%%%%%%%%%%%%%%%%%%%%%
Next, we define a set of edges $E_+ \subseteq E$ by
\[
E_+ = \{(S, T)\ |\ \Delta(S\backslash\{g(S,T)\}, g(S,T)) \geq \alpha\}
\]
This is the set of ``special edges" alluded to in Section~\ref{sec:techniques}.

The informal argument for the next lemma is as follows. By Lemma~\ref{lem:graph-symmetric}, we have $(S,T) \in E$ if and only if $(T,S) \in E$. Then Lemma~\ref{lem:marginal-value} (combined with submodularity) implies that at least one of $\Delta(S\backslash \{g(S,T)\}, g(S,T)) \geq \alpha$ and $\Delta(T\backslash \{g(S,T)\}, g(S,T)) \geq \alpha$ is true, so at least one of $(S,T)$ and $(T,S)$ must be in $E_+$.

\begin{restatable}{lemma}{lemEPlusSize}
\label{lem:e-plus-size}
Assuming Condition~\ref{cut-failed-cond}, $|E_+| \geq |E|/2$.
\end{restatable}

\begin{proof}
Let $(S,T)$ be some edge in $E$: then by Lemma~\ref{lem:graph-symmetric}, $(T,S) \in E$. It suffices to show that for every edge $(S,T) \in E$, at least one of $(S,T)$ and $(T,S)$ are in $E_+$. Assume $(S,T)\not\in E_+$: otherwise we are done. Then
\[
\Delta\big(S\backslash\{g(S,T)\}, g(S,T)\big) < \alpha
\]
Thus by Lemma~\ref{lem:marginal-value},
\[
\Delta\big(\mss, g(S,T)\big) \geq \alpha
\]
Since $(T,S)$ is an edge in the graph, $S\subseteq (M\backslash T)\cup\{g(S,T)\}$. Therefore $S\backslash\{g(S,T)\} \subseteq M\backslash T$. Thus by submodularity, $\Delta(S\backslash\{g(S,T)\}, g(S,T)) \geq \Delta(M\backslash T, g(S,T)) \geq \alpha$. Therefore $(S,T) \in E_+$.
\end{proof}

%%%%%%%%%%%%%%%%%%%%%%%%%%%%%%%%%%%%%%%%

Lemma~\ref{lem:counting-bound} follows from a simple counting argument.

\begin{restatable}{lemma}{lemCountingBound}
\label{lem:counting-bound}
For any integers $m$ and $\ell$, $\sum\limits_{j=0}^\ell \mbinom{m}{j} \leq (m+1)^\ell$.
\end{restatable}

\begin{proof}
The left-hand-side is number of subsets of $[m]$ of size at most $\ell$. The right-hand-side is the number of ways to select $\ell$ elements from $[m]\cup\{d\}$, where each element can be selected multiple times, and including ordering. We think of $d$ as a dummy element. For each subset $S\subseteq [m]$ counted by $\sum_{j=0}^\ell \binom{m}{j}$, we represent it in $(m+1)^\ell$ as follows: first select element $d$ $\ell - |S|$ times, and then select the elements in $S$ in any order. Thus each subset of $[m]$ counted by the left-hand-side is represented in a unique way by the right-hand-side, and so $\sum_{j=0}^\ell \binom{m}{j} \leq (m+1)^\ell$.
\end{proof}

%%%%%%%%%%%%%%%%%%%%%%%%%%%%%%%%%%%%%%%%
%%%%%%%%%%%%%%%%%%%%%%%%%%%%%%%%%%%%%%%%
% FINAL THEOREM

We are now ready to prove the final theorem. Recall the following definitions from Section~\ref{sec:prop-n=2-submod-upper}:
\begin{align*}
S_{\le k} =&\ \{g_j \in S\ |\ j \leq k\}\\
\delta_k^S =&\ v_1(S_{\le k}) - v_1(S_{\le k-1})
\end{align*}

\EFsubmodUpper*

\begin{proof}
Correctness of Protocol~\ref{pro:minimal-upper} follows from Lemma~\ref{lem:minimal-correct}, so it remains only to bound the communication cost.

We prove that the number of minimal bundles is (strictly) less than $2(m+1)^{\frac{8}{1-c}} = 2(m+1)^{4/\alpha}$, assuming Condition~\ref{cut-failed-cond}. Let $\beta = 4/\alpha$, and suppose that the number of minimal bundles is at least $2(m+1)^{4/\alpha} = 2(m+1)^\beta$. By Lemma~\ref{lem:counting-bound}, the number of minimal bundles of size at most $\beta $ is at most $(m+1)^{\beta}$. Thus there are at least $(m+1)^{\beta}$ minimal bundles $S$ where $|S| > \beta$.

So at least half of the minimal bundles have size more than $\beta$. Let $G = (V,E)$ be the minimal bundle graph. Then by Lemma~\ref{lem:graph-outdegree}, at least half of the minimal bundles in $V$ have out-degree more than $\beta$. Therefore $|E| > \beta |V|/2$. Then by Lemma~\ref{lem:e-plus-size}, $|E_+| > \beta |V|/4 = |V| /\alpha$.

For a bundle $S$, let $X_+^S$ be the set of out-edges from $S$ that are in $E_+$. Formally,
\[
X_+^S = \{(S,T) \in E\ |\ \Delta(S\backslash\{g(S,T)\}, g(S,T)) \geq \alpha\}
\]
and we can define the corresponding goods by $g(X_+^S) =\{g \in S\ |\ \Delta(S\backslash\{g\}, g) \geq \alpha\}$.

We next show that there must exist a minimal bundle $S \in V$ where $|X_+^S| > 1/\alpha$. Suppose that $|X_+^S| \le 1/\alpha$ for all $S \in V$: then
\[
|E_+| \leq |V|/\alpha
\]
which contradicts $|E_+| > |V|/\alpha$. Therefore there exists some bundle $S$ with $|X_+^S| > 1/\alpha$. By definitions, we have
\[
v_1(S) = \sum\limits_{k = 1}^m \delta_k^S = \sum\limits_{k: g_k \in S} \delta_k^S
= \sum\limits_{k: g_k \in S} \Delta(S_{\le k-1}, g_k)
\geq \sum\limits_{k: g_k \in g(X_+^S)} \Delta(S_{\le k-1}, g_k)
\]
Because $S_{\le k-1} \subseteq S$ and $g_k \not\in S_{\le k-1}$, we have $S_{\le k-1} \subseteq S\backslash\{g_k\}$. Therefore by submodularity,\footnote{This is the crucial use of submodularity: that we can add in the items in $S$ one by one, and the value of the set increases by at least $\Delta(S\backslash\{g_k\}, g_k)$ each time. This allows us to pump the value of $S$ over $v_1(M)$.}
\[
\sum\limits_{k:g_k \in g(X_+^S)} \Delta(S_{\le k-1}, g_k)
\geq \sum\limits_{k: g_k \in g(X_+^S)} \Delta(S\backslash\{g_k\}, g_k) 
\geq \sum\limits_{k: g_k \in g(X_+^S)} \alpha
= \alpha|X_+^S| > 1
\]
But $v_1(M) = 1$, so this is a contradiction. Therefore the number of minimal bundles is less than $2(m+1)^{\frac{8}{1-c}}$. 

Thus the number of minimal bundles is at most $2(m+1)^{\frac{8}{1-c}} - 1$. If the protocol terminates in step 1, just one bundle is communicated (and zero values), so the communication cost bound is trivially satisfied. Suppose the protocol does not terminate in step 1: then player 1 sends at most $2(m+1)^{\frac{8}{1-c}} - 1$ minimal bundles, as well as $S^*(v_1)$. Thus at most $2(m+1)^{\frac{8}{1-c}}$ bundles are sent, each of which require $m$ bits to communicate.

Player 1 also sends $\bfc_i(S^*(v_1))$. By definition of $\bfc_i^{EF}$, $\bfc_i(S^*(v_1))$ can be expressed as the ratio of two values, each of which takes $\valsize$ bits to communicate. Therefore the total communication cost is
\[
2m(m+1)^{\frac{8}{1-c}} + 2\valsize
\]
as required.
\end{proof}

We will show formally in Section~\ref{sec:EF-n=2-submod-lower} that Theorem~\ref{thm:EF-n=2-submod-upper} is tight, meaning that exponential communication can be required when $c=1$. To see why the minimal bundle argument fails for $c=1$, consider an additive (and hence submodular) valuation over an even number of items, where the value of each item is one. Then a bundle is minimal if and only if it contains exactly half the items, and there are exponential number of such bundles.

\section{Lower bound approach}\label{sec:lower-bound-approach}

In Section~\ref{sec:EF-n=2-submod-lower}, we will prove a lower bound that matches the PAS from Section~\ref{sec:EF-n=2-submod-upper}. Before we do that, we describe our general lower bound approach in this section. All of our lower bounds will rely on reductions from two well-known problems in communication complexity: determining whether two bit strings are equal, and determining whether two bit strings are disjoint. Let $x_i$ denote the bit string held by player $i$, and let $x_{ij}$ denote the $j$th bit of $x_i$. An input $(x_1, x_2)$ is a no-instance of the \equ problem if and only if there exists $j$ where $x_{1j} \neq x_{2j}$. An input $(x_1, x_2)$ is a no-instance of the \disj problem if and only if there exists $j$ where $x_{1j} = x_{2j} = 1$. The following lemma states that \disj is hard in the randomized setting (and thus also in the deterministic setting).

\begin{lemma}[\cite{Kalya92:Probabilistic, Razborov92:Distributional}]\label{lem:disjoint-lower}
Any randomized protocol which solves \disj for bit strings of length $\ell$ has communication cost $\Omega(\ell)$.
\end{lemma}

The following well-known lemma states that \equ is hard in the deterministic setting.

\begin{lemma}\label{lem:equality-lower}
Any deterministic protocol which solves \equ for bit strings of length $\ell$ has communication cost at least $\ell$.
\end{lemma}

Perhaps surprisingly, \equ admits a constant communication randomized protocol, due to \cite{Yao79:Complexity}.
\begin{lemma}[\cite{Yao79:Complexity}]\label{lem:equality-upper}
There exists a randomized protocol which solves \equ and has communication cost $O(1)$.
\end{lemma}
\noindent The protocol for Lemma~\ref{lem:equality-upper} asks each player to compute the inner product mod 2 of her bit string and a random string. The protocol then compares those inner products. The Principle of Deferred Decisions can be used to show that this protocol arrives at the correct answer with probability at least 75\%. Lemma~\ref{lem:equality-upper} will be a key element of our randomized upper bound in Section~\ref{sec:n=2-rand-upper}.

All of our lower bounds have the following structure. Given two bit strings $x_1$ and $x_2$ of length $\ell = \Omega(\binom{2k}{k})$, we construct a corresponding instance of \textsc{Fair Division} with $O(k)$ items. In the two player case, each index in the bit strings will correspond to a possible allocation that gives each player $k$ items.

Our constructed instance will have that a property that a $c$-$P$ allocation exists if and only if $(x_1, x_2)$ is a no-instance\footnote{Note that no-instances of \equ or \disj become instances where a $c$-$P$ allocation \emph{does} exist.} of \equ (for a deterministic lower bound), or a no-instance of \disj (for a randomized lower bound). Thus if there existed a protocol for \textsc{Fair Division} with communication cost less than $\Omega(\binom{2k}{k})$, it could also be used to solve \equ or \disj in communication less than $\Omega(\ell)$. This is impossible according to Lemmas~\ref{lem:disjoint-lower} and \ref{lem:equality-lower}, so any protocol for \textsc{Fair Division} requires exponential communication. 

Using this framework, all that is needed to prove a lower bound for a particular set of parameters (property $P$, constant $c$, and a valuation class) is:
\begin{enumerate}
\item Given bit strings $x_1$ and $x_2$ of length $\Omega(\binom{2k}{k})$, define how to construct a corresponding instance of \textsc{Fair Division} with $O(k)$ items.
\item Show that a $c$-$P$ allocation exists in the constructed instance if and only if $(x_1, x_2)$ is a no-instance of \equ or \disj.
\item Show that the valuations in the constructed instance of \textsc{Fair Division} are of the desired valuation class.
\end{enumerate}

%%%%%%%%%%%%%%%%%%%%%%%%%%%%%%%%
More specifically, our \fd instance will have two players and $2k$ items. Valuations will be constructed such that a player will never be happy if she receives fewer than $k$ items, so both players will have to receive exactly $k$ items. There are $\binom{2k}{k}$ allocations which give each player $k$ items, and this gives rise to the exponential communication lower bound.

In fact, we can do this in a very standardized way for the two player deterministic case. Given bit strings of length $\frac{1}{2}\binom{2k}{k}$, we define a list of allocations $\T = (T_1, T_2...T_{|\T|})$ where each $T_j = (T_{j1}, T_{j2}) \in \T$ is an allocation giving each player $k$ items: $|T_{j1}| = |T_{j2}| = k$. It is important that $\T$ does not contain every such allocation: in particular, for any allocation $A \in \T$, $\overline{A} \not\in \T$.\footnote{Recall that for $A = (A_1, A_2)$, $\overline{A} = (A_2, A_1)$.} The allocations in $\T$ appear in an arbitrary (but known) order. Note that $|\T| = \frac{1}{2}\binom{2k}{k}$.

Lemma~\ref{lem:EF-equality-reduction} states that under this approach, all that is necessary to complete the lower bound is to construct valuations satisfying three particular properties. The exact way valuations are constructed will depend on what class we wish them to belong to (general, subadditive, or submodular). We only prove the lemma for the $c$-EF in the two player deterministic setting. A similar result is possible for other settings, but this is only setting where we prove enough different lower bounds to make it worth having a separate lemma.

For a bit string $x_i$, let $\overline{x_i}$ denote the string obtained by flipping every bit: $x_{ij} \neq \overline{x_i}_j$ for all $j$. We will define two new bit strings, $y_1$ and $y_2$, by $y_1 = x_1$ and $y_2 = \overline{x_2}$. Also, recall that for a player $i$, $\overline{i}$ denotes the other player.

The lemma relies on three conditions. Condition~\ref{reduction-cond1} states that neither player is happy with any bundle containing fewer than $k$ items: then any $c$-$P$ allocation must either be $A$ or $\overline{A}$ for some $A \in \T$. Condition~\ref{reduction-cond2} states that player $i$ is unhappy receiving $T_{j\overi}$ when $y_{ij} = 1$ (and happy receiving $T_{ji}$). Condition~\ref{reduction-cond3} states that player $i$ is unhappy receiving $T_{ji}$ when $y_{ij} = 0$ (and happy receiving $T_{j\overi}$). Thus we want to find an index $j$ where either $y_{1j} = y_{2j} = 1$, in which case the allocation $(T_{j1}, T_{j2})$ is $c$-$P$, or where $y_{1j} = y_{2j} = 0$, in which case the allocation $(T_{j2}, T_{j1})$ is $c$-$P$. Therefore we are looking for an index where $y_{1j} = y_{2j}$, which is equivalent to $x_{1j} \neq x_{2j}$. This is exactly the \equ problem.

%%%%%%%%%%%%%%%%%%%%%%%%%%%%%%%%
\begin{lemma}
\label{lem:EF-equality-reduction}
Given bit strings of length $\mfrac{1}{2}\mbinom{2k}{k}$ for some integer $k$, let $M = [2k]$ and $N=[2]$. Let $y_1 = x_1$ and $y_2 = \overline{x_2}$, and let $c$ be some constant. Let $\T = (T_1, T_2...T_{|\T|})$ be a list of allocations as described above. Suppose $v_1, v_2$ can be constructed such that the following conditions are met:
\begin{restatable}{condition}{firstReductionCond}\label{reduction-cond1}
For all $|S| < k$ and both $i$, $v_i(S) < c \cdot v_i(\mss)$.
\end{restatable}
\vspace{-.1 in}
\begin{restatable}{condition}{secondReductionCond}\label{reduction-cond2}
Whenever $y_{ij} = 1$, $v_i(T_{j\overi}) < c\cdot v_i(T_{ji})$.
\end{restatable}
\vspace{-.1 in}
\begin{restatable}{condition}{thirdReductionCond}\label{reduction-cond3}
Whenever $y_{ij} = 0$, $v_i(T_{ji}) < c\cdot v_i(T_{j\overi})$.
\end{restatable}
Then any deterministic protocol which finds a $c$-EF allocation for two players requires exponential communication. Specifically,
\[
D(2, 2k, \text{EF}, c) \geq \frac{1}{2} \binom{2k}{k}
\]
\end{lemma}

\begin{proof}
We reduce from \equ. Given bit strings $x_1$ and $x_2$ of length $\mfrac{1}{2}\mbinom{2k}{k}$ for some integer $k$, we construct the following instance of \fd. Let $N, M, (y_1, y_2)$, and $\T$ be as defined in the statement of Lemma~\ref{lem:EF-equality-reduction}. Also assume that $v_1$ and $v_2$ satisfy Conditions~\ref{reduction-cond1}, \ref{reduction-cond2}, and \ref{reduction-cond3}.

Suppose that $(x_1, x_2)$ is a no-instance of \equ: then there exists $j$ where $x_{1j} \neq x_{2j}$. Therefore $y_{1j} = y_{2j}$. If $y_{1j} = y_{2j} = 1$, then by Condition~\ref{reduction-cond2},
\[
v_i(T_{ji}) > \frac{1}{c} v_i(T_{j\overi}) \geq c \cdot v_i(T_{j\overi})
\]
for both $i$. Thus the allocation $T_j$ is $c$-EF, because each player $i$ receives $T_{ji}$. If $y_{1j} = y_{2j} = 0$, then by Condition~\ref{reduction-cond3},
\[
v_i(T_{j\overi}) > \frac{1}{c} v_i(T_{ji}) \geq c \cdot v_i(T_{ji})
\]
for both $i$. Thus the allocation $\overline{T_j}$ is $c$-EF, because each player $i$ receives $T_{j\overi}$. Therefore if $(x_1, x_2)$ is a no-instance of \equ, there exists an allocation satisfying $c$-EF.

Suppose that $(x_1, x_2)$ is a yes-instance of \equ: then for every $j$, $y_{1j} \neq y_{2j}$. For any allocation $A$ where $|A_i | < k$ for some $i$, we have $v_i(A_i) < c\cdot v_i(A_{\overi})$ by Condition~\ref{reduction-cond1}. Thus $A$ cannot be $c$-EF whenever $|A_i | < k$ for some $i$.

Now consider an arbitrary allocation $A$ where $|A_1| = |A_2| = k$. For any such allocation, there must exist $j$ where either $A = T_j$, or $A = \overline{T_j}$. Since $y_{1j} \neq y_{2j}$, there exists a player $i$ where $y_{ij} = 0$, and $y_{\overline{i}j} = 1$. Then by Condition~\ref{reduction-cond3}, $v_i(T_{ji}) < c\cdot v_i(T_{j\overi})$. Also, $v_{\overi}(T_{j\overline{\overi}}) < c\cdot v_{\overi}(T_{j\overi})$ by Condition~\ref{reduction-cond2}, where $\overline{\overi} = i$ represents the player other than $\overi$. Thus $v_{\overi}(T_{ji}) < c\cdot v_{\overi}(T_{j\overi})$.

Therefore neither player is happy with bundle $T_{ji}$. But since either $A = T_j$ or $A = \overline{T_j}$, there must be a player who receives $T_{ji}$, is hence is not happy. Thus no allocation where $|A_1 | = |A_2| = k$ can be $c$-EF, no allocation is $c$-EF.
\end{proof}

This lemma will be useful in a variety of settings. In the next section, we will use this lemma to prove a lower bound for 1-EF that matches the PAS from Section~\ref{sec:EF-n=2-submod-upper}.

%This lemma will be useful in a variety of settings. As an example, for general valuations, we can define $v_i(S)$ by
%\[ 
%v_i(S) =
%   \begin{cases}    
%      0 & \textnormal{if}\ \ |S| < k\\     
%      1 & \textnormal{if}\ \ |S| > k\\  
%      1 & \textnormal{if}\ \ \exists j\ \ S = T_{ji} \ \ \textnormal{where}\ \ y_{ij} = 1\\
%      1 & \textnormal{if}\ \ \exists j\ \ S = T_{j\overline{i}} \ \ \textnormal{where}\ \ y_{ij} = 0\\      
%      0 & \textnormal{if}\ \ \exists j\ \ S = T_{ji} \ \ \textnormal{where}\ \ y_{ij} = 0\\
%      0 & \textnormal{if}\ \ \exists j\ \ S = T_{j\overline{i}} \ \ \textnormal{where}\ \ y_{ij} = 1\\      
%   \end{cases}
%\]
%The formal proof that this valuation satisfies the three necessary conditions appears in the appendix (along with the other formal applications of this lemma).
%
%However, there are some settings where both Lemma~\ref{lem:EF-equality-reduction} and the more general approach described in this section break down. The remainder of the main body of the paper discusses two settings which (perhaps surprisingly) admit efficient deterministic protocols.

\section{1-EF is hard for submodular valuations}\label{sec:EF-n=2-submod-lower}

In this section, we use the general approach described in Section~\ref{sec:lower-bound-approach} to show that 1-EF requires exponential communication, even for two players with submodular valuations. This shows that the PAS for this setting from Section~\ref{sec:EF-n=2-submod-upper} is the best we can hope for.

Formally, Section~\ref{sec:EF-n=2-submod-upper} showed that $D_{submod}(2,m,\text{EF}, c)$ is polynomial in $m$ when $c < 1$. We now show that $D_{submod}(2,m,\text{EF}, c)$ is exponential when $c = 1$. Section~\ref{sec:prop-n=2-submod-upper} showed that $D_{submod}(2,m,\text{Prop}, c)$ is polynomial for any $c$, so there is no lower bound necessary there. Thus this section resolves the deterministic submodular case for two players.

%Recall Lemma~\ref{lem:EF-equality-reduction}, which gives a standardized way to prove deterministic lower bounds for EF for two players. Starting from bit strings $x_1, x_2$, a \fd instance is constructed with $M = [2k]$, $N = [2]$, $y_1 = x_1$, and $y_2 = \overline{x_2}$. Recall that a list of allocations $\T = (T_1, T_2...)$ is defined where each $T_j = (T_{j1}, T_{j2})$ is an allocation giving each player $k$ items. Also, for every such allocation $A$, exactly one of $A$ and $\overline{A}$ appears in $\T$. Then all that is needed to complete the reduction is to show how to construct valuations $v_1, v_2$ such that the following conditions are satisfied:
%
%\firstReductionCond*
%\secondReductionCond*
%\thirdReductionCond*

%%%%%%%%%%%%%%%%%%%%%%%%%%%%%%%%%%%%%%%

\begin{theorem}\label{thm:EF-n=2-submod-lower}
For two players with submodular valuations, any deterministic protocol which determines whether a $1$-EF allocation exists requires an exponential amount of communication. Specifically,
\[
D_{submod}(2,2k,\text{EF},1) \geq \frac{1}{2}\binom{2k}{k}
\]
\end{theorem}

\begin{proof}
Given bit strings of length $\mfrac{1}{2}\mbinom{2k}{k}$ for some integer $k$, define $M, N, (y_1, y_2)$, and $\T$ as in Lemma~\ref{lem:EF-equality-reduction}. We need only to construct submodular valuations $v_1, v_2$ such that Conditions~\ref{reduction-cond1}, \ref{reduction-cond2}, and \ref{reduction-cond3} are met. We define each $v_i$ by
\[ v_i(S) =
   \begin{cases} 
      3|S| & \textnormal{if}\ \ |S| < k\\
      3k & \textnormal{if}\ \ |S| > k\\
      3k & \textnormal{if}\ \ S = T_{ji} \ \ \textnormal{and}\ \ y_{ij} = 1\\
      3k & \textnormal{if}\ \ S = T_{j\overi} \ \ \textnormal{and}\ \ y_{ij} = 0\\
      3k - 1 & \textnormal{if}\ \ S = T_{ji} \ \ \textnormal{and}\ \ y_{ij} = 0\\
      3k - 1& \textnormal{if}\ \ S = T_{j\overi} \ \ \textnormal{and}\ \ y_{ij} = 1
   \end{cases}
\]
Importantly, for every bundle $S$ with $|S| = k$, there exists exactly one pair $(i,j)$ where $S = T_{ji}$. Thus if $|S| = k$, $S$ falls under exactly one of the last four cases in the definition of $v_i$.

If $|S| < k$, we have $|\mss| > k$, so $v_i(S) < 3k = v_i(\mss)$. This satisfies Condition~\ref{reduction-cond1}. Suppose $y_{ij} = 1$ for some $i,j$: then $v_i(T_{j\overi}) = 3k - 1 < 3k = v_i(T_{ji})$, so Condition~\ref{reduction-cond2} is satisfied. Suppose $y_{ij} = 0$ for some $i,j$: then similarly, $v_i(T_{ji}) = 3k - 1 < 3k = v_i(T_{j\overi})$. Thus Condition~\ref{reduction-cond3} is satisfied as well.

It remains to show that the valuations are submodular. To do this, we examine $v_i(S \cup \{g\}) - v_i(S)$, for any bundle $S$ and item $g\not\in S$.
\[ v_i(S\cup\{g\}) - v_i(S) =
   \begin{cases} 
      3 & \textnormal{if}\ \ |S\cup\{g\}| < k\\
      2\ \textnormal{or}\ 3 & \textnormal{if}\ \ |S \cup\{g\}| = k\\
      0\ \textnormal{or}\ 1 & \textnormal{if}\ \ |S\cup\{g\}| = k + 1\\
      0 & \textnormal{if}\ \ |S\cup\{g\}| > k + 1      
   \end{cases}
\]
Therefore $v_i(S\cup\{g\}) - v_i(S)$ is non-increasing with $|S|$. Thus $v_i(X\cup\{g\}) - v_i(X) \geq v_i(Y\cup\{g\}) - v_i(Y)$ whenever $|X| < |Y|$. If $X \subseteq Y$, either $|X| < |Y|$ or $X = Y$. When $X = Y$, we trivially have $v_i(X\cup\{g\}) - v_i(X) = v_i(Y\cup\{g\}) - v_i(Y)$. Thus we have $v_i(X\cup\{g\}) - v_i(X) \geq v_i(Y\cup\{g\}) - v_i(Y)$ whenever $X\subseteq Y$, and so $v_i$ is submodular.

\end{proof}

Recall that Section~\ref{sec:EF-n=2-submod-upper} gave a PAS for this setting, where for any fixed $c$, communication at most $2(m+1)^{\frac{8}{1-c}}$ is required. In a fully polynomial-communication approximation scheme (FPAS), the dependence in $\frac{1}{1-c}$ is required to be polynomial. The PAS from Section~\ref{sec:EF-n=2-submod-upper} is not an FPAS, since the dependence on $\frac{1}{1-c}$ is exponential.

The above proof of Theorem~\ref{thm:EF-n=2-submod-lower} actually shows that for any $c > \frac{3k-1}{3k} = \frac{3m-2}{3m}$, exponential communication is required. This does not contradict the PAS from Section~\ref{sec:EF-n=2-submod-upper}, because $\frac{3m-2}{3m}$ is not a fixed constant (it depends on $m$). However, this does rule out the possibility of an FPAS. To see this, suppose an FPAS existed, and consider some $c > \frac{3m-2}{3m}$. Then the FPAS would have communication cost polynomial in $\frac{1}{1 - c}$. We have
\[
\frac{1}{1-c} > \frac{1}{1 - \mfrac{3m-2}{3m}} = \frac{3m}{2}
\]
so the communication cost is polynomial of $m$. But the proof of Theorem~\ref{thm:EF-n=2-submod-lower} shows communication exponential in $m$ is required, which is a contradiction.

Finally, we note that the proof of Theorem~\ref{thm:EF-n=2-submod-lower} can easily be adapted to prove exponential lower bounds on the communication complexity of maximizing Nash welfare (the product of player utilities) or egalitarian welfare (the minimum player utility).

\section{Everything is hard for more than two players}\label{sec:n>2}

In this section, we show that \fd requires an exponential amount of communication whenever there are more than two players: even when randomization is allowed, even for submodular valuations, and for any $c > 0$. This will allow us to focus on the two player setting for the rest of the paper.

Before proving the theorems, we discuss the multiparty (i.e., $n>2$) communication complexity model. As mentioned in Section~\ref{sec:model}, there is more than one such model. This will turn out not to matter in our setting. The reason is that our lower bounds will hold even when only player 1 and player 2 have private valuations, and the valuations of all other players are public information. One can think of the other players as not really being agents, and just being a (publicly known) part of the input. Thus we never actually consider multiparty communication. In this way, the theorem that we are really proving is that when there are more than two \fd players, the problem is hard in the two-party communication complexity model.

We first prove hardness for envy-freeness, and then reduce envy-freeness to proportionality. Recall that \disj has randomized communication complexity $\Omega(\ell)$, where $\ell$ is the length of the bit strings (Lemma~\ref{lem:disjoint-lower}).

\begin{theorem}\label{thm:EF-n>2}
For any $n>2$ and any $c>0$, any randomized protocol which determines whether a $c$-EF allocation exists requires an exponential amount of communication, even for submodular valuations. Specifically,
\[
R_{submod}(n,2k+n-2,\text{EF},c) \in \Omega \left(\binom{2k}{k}\right)
\]
for any $n>2$ and $c > 0$.
\end{theorem}

\begin{proof}
We reduce from \disj . Given bit strings $x_1$ and $x_2$ of length $\mbinom{2k}{k}$, we construct a fair division instance as follows. Although there will be more than two players, there are only two bit strings. Let player 1 hold $x_1$ and player 2 hold $x_2$, and the other players will have no bit strings.

Let $M_1 = [2k], M_2 = \{g_3...g_n\}$, and $M = M_1 \cup M_2$: note that $|M| = 2k + n -2$.  Let $N = [n]$. We define a similar list of allocations $\T = (T_1, T_2 ...)$, where $T_j = (T_{j1}, T_{j2})$. Here each $T_j$ is an allocation over only $M_1$, and for just two players. Any such allocation $A$ where $|A_1 | = |A_2| = k$ is in $\T$ (and so is $\overline{A}$). Note that $|\T| = \mbinom{2k}{k}$. For $i \in \{1,2\}$, $v_i$ is given by
\[ v_i(S) =
   \begin{cases} 
      k & \textnormal{if}\ \ g_3 \in S\\      
      |S|c & \textnormal{if}\ \ |S| < k\ \ \textnormal{and}\ \ g_3 \not\in S\\
      kc & \textnormal{if}\ \ |S| > k\ \ \textnormal{and}\ \ g_3 \not\in S\\
      (k-\frac{1}{2})c & \textnormal{if}\ \ |S| = k \ \ \textnormal{and}\ \ g_3\not\in S
      \ \ \textnormal{and}\ \ S \cap M_2\ne \emptyset\\ 
      kc & \textnormal{if}\ \ \exists j\ \ S = T_{ji} \ \ \textnormal{where}\ \ x_{ij} = 1\ \ \textnormal{and}\ \ g_3 \not\in S\\  
      (k - \frac{1}{2})c & \textnormal{if}\ \ \exists j\ \ S = T_{ji} \ \ \textnormal{where}\ \ x_{ij} = 0\ \ \textnormal{and}\ \ g_3 \not\in S  
   \end{cases}
\]
Every allocation giving each player $k$ items occurs in $\T$ exactly once. Thus when $|S| = k$ and $S \subset M_1$, exactly one of the last two cases occurs, and any such $j$ must be unique. For $i > 2$, $v_i(S)$ is given by
\[ v_i(S) =
   \begin{cases} 
      1 & \textnormal{if}\ \ g_i \in S\\
      0 & \textnormal{otherwise}\\
   \end{cases}
\]
Suppose that $(x_1, x_2)$ is a no-instance of \disj: then there exists $j$ where $x_{1j} = x_{2j} = 1$. Consider the allocation $A$ where $A_i = T_{ji}$ for $i \leq 2$, and $A_i = \{g_i\}$ for $i > 2$. For $i>2$, $v_i(A_i) = 1$ and $v_i(A_{i'}) = 0$ for all $i' \neq i$, so each player $i > 2$ is happy. For $i \leq 2$, we have $v_i(A_i) = v_i(T_{ji}) = kc$, and $v_i(A_{i'}) \leq k$ for all $i'$. Therefore for all $i, i'$, $v_i(A_i) \geq c v_i(A_{i'})$, so $A$ is $c$-EF.

Suppose that $(x_1, x_2)$ is a yes-instance of \disj: then for every $j$, there exists $i$ where $x_{ij} = 0$. Suppose that a $c$-Prop allocation $A = (A_1, A_2)$ exists. We first claim that for every $i > 2$, $g_i \in A_i$: if not, $v_i(A_i) = 0$, so player $i$ will envy whichever player receives $g_i$. 

Thus for $i \le 2$, 
\[
v_i(A_i) \geq c\cdot v_i(A_3) \geq c\cdot v_i(\{g_3\}) = kc
\]
Suppose a player $i \le 2$ receives strictly fewer than $k$ items in $A_i$: then $v_i(A_i) < kc$, since none of those items can be $g_3$. This is a contradiction, so we have $|A_1 | = |A_2 | = k$. Since $\T$ contains all of the allocations which give each player $k$ items, there must exist $j$ where $A_i = T_{ji}$ for both $i$, and $v_i(A_i) \geq kc$. But that implies that $x_{1j} = x_{2j} = 1$, which is a contradiction. Therefore no allocation is $c$-Prop.

It remains to show that the valuations are submodular. For $i>2$, $v_i$ is trivially submodular. We now we examine $v_i(S \cup \{g\}) - v_i(S)$ for $i \le 2$, any bundle $S$, and any item $g\not\in S$ where $g_3 \not\in  S\cup \{g\}$.
\[ v_i(S\cup\{g\}) - v_i(S) =
   \begin{cases} 
      c & \textnormal{if}\ \ |S\cup\{g\}| < k\\
      c\ \textnormal{or}\ c/2 & \textnormal{if}\ \ |S \cup\{g\}| = k\\
      c/2\ \textnormal{or}\ 0 & \textnormal{if}\ \ |S\cup\{g\}| = k + 1\\
      0 & \textnormal{if}\ \ |S\cup\{g\}| > k + 1      
   \end{cases}
\]
Therefore $v_i(S\cup\{g\}) - v_i(S)$ is non-increasing with $|S|$  when $g_3 \not \in S\cup\{g\}$. Thus $v_i(X\cup\{g\}) - v_i(X) \geq v_i(Y\cup\{g\}) - v_i(Y)$ whenever $|X| < |Y|$ and $g_3 \not \in S\cup\{g\}$. If $X \subseteq Y$, either $|X| < |Y|$ or $X = Y$. When $X = Y$, we trivially have $v_i(X\cup\{g\}) - v_i(X) = v_i(Y\cup\{g\}) - v_i(Y)$. Thus we have $v_i(X\cup\{g\}) - v_i(X) \geq v_i(Y\cup\{g\}) - v_i(Y)$ whenever $X\subseteq Y$ and $g_3 \not \in S\cup\{g\}$. Therefore the submodularity condition is satisfied when $g_3 \not \in S\cup\{g\}$.

There are two remaining cases: when $g_3 \in S$, or when $g = g_3$. For $g_3 \in S$, $v_i(S\cup\{g\}) - v_i(S) = 0$ for all $S$ and $g$, so the condition is satisfied in this case. For $g = g_3$, we have $v_i(X\cup \{g_3\}) - v_i(X) = v_i(M) - v_i(X)$ and $v_i(Y\cup \{g_3\}) - v_i(Y) = v_i(M) - v_i(Y)$. If $X\subseteq Y$, we have $v_i(X) \leq v_i(Y)$, so $v_i(X\cup \{g_3\}) - v_i(X) \geq  v_i(Y\cup \{g_3\}) - v_i(Y)$. Therefore $v_i$ is submodular for all $i$.
\end{proof}

%%%%%%%%%%%%%%%%%%%%%%%%%%%%%%%%%%%%%%%%%

We now prove hardness for proportionality for more than two players, by reducing from envy-freeness.

\begin{theorem}\label{thm:prop-n>2}
For any $n>2$ and any $c>0$, any randomized protocol which determines whether a $c$-Prop allocation exists requires an exponential amount of communication, even for submodular valuations. Specifically,
\[
R_{submod}(n,2k+n-2,\text{Prop},c) \in \Omega \left(\binom{2k}{k}\right)
\]
for any $c > 0$.
\end{theorem}

\begin{proof}
We reduce from \disj. Given an input $(x_1, x_2)$, we define $v_i$ as in the proof of Theorem~\ref{thm:EF-n>2}, except using $c/n$ instead of $c$. That is, for $i \le 2$,
\[ v_i(S) =
   \begin{cases} 
      k & \textnormal{if}\ \ g_3 \in S\\      
      |S|c/n & \textnormal{if}\ \ |S| < k\ \ \textnormal{and}\ \ g_3 \not\in S\\
      kc/n & \textnormal{if}\ \ |S| > k\ \ \textnormal{and}\ \ g_3 \not\in S\\
      (k-\frac{1}{2})c/n & \textnormal{if}\ \ |S| = k \ \ \textnormal{and}\ \ g_3\not\in S
      \ \ \textnormal{and}\ \ S \cap M_2\ne \emptyset\\ 
      kc/n & \textnormal{if}\ \ \exists j\ \ S = T_{ji} \ \ \textnormal{where}\ \ x_{ij} = 1\ \ \textnormal{and}\ \ g_3 \not\in S\\  
      (k - \frac{1}{2})c/n & \textnormal{if}\ \ \exists j\ \ S = T_{ji} \ \ \textnormal{where}\ \ x_{ij} = 0\ \ \textnormal{and}\ \ g_3 \not\in S  
   \end{cases}
\]

%\[ v_i(S) =
%   \begin{cases} 
%      k & \textnormal{if}\ \ g_3 \in S\\      
%      |S|c/n & \textnormal{if}\ \ |S| < k\\
%      kc/n & \textnormal{if}\ \ |S| > k\\  
%      kc/n & \textnormal{if}\ \ \exists j\ \ S = T_{ji} \ \ \textnormal{where}\ \ x_{ij} = 1\\  
%      (k - \frac{1}{2})c/n & \textnormal{if}\ \ \exists j\ \ S = T_{ji} \ \ \textnormal{where}\ \ x_{ij} = 0 \\
%      v_i(S\backslash M_2) & \textnormal{if}\ \ S \cap M_2\ne \emptyset \ \ \textnormal{and}\ \ g_3 \not\in S          
%   \end{cases}
%\]
It was shown in the proof of Theorem~\ref{thm:EF-n>2} that these valuations are submodular.

Theorem~\ref{thm:EF-n>2} implies that $\Omega(\binom{2k}{k})$ communication is required to determine whether a $\frac{c}{n}$-EF allocation exists under these valuations. We will show that under these valuations, an allocation is $c$-Prop if and only if it is $\frac{c}{n}$-EF. This will imply that determining whether a $c$-Prop allocation exists is just as hard as whether a $\frac{c}{n}$-EF allocation exists.

In order for an allocation $A$ to be $\frac{c}{n}$-EF or $c$-Prop, we must have $v_i(A_i) > 0$ for all $i$. Thus assume $g_i \in A_i$ for all $i > 2$, and we need only consider $i \leq 2$.

Suppose an allocation $A$ is $c$-Prop: then for $i \le 2$, $v_i(A_i) \geq \mfrac{c}{n}v_i(M) = \mfrac{kc}{n}$. Since $v_i(A_{i'}) \leq k$ for all $i'$, we have
\[
v_i(A_i) \geq \frac{kc}{n} \geq \frac{c}{n}v_i(A_{i'})
\]
for all $i'$. Therefore $A$ is $\frac{c}{n}$-EF.

Suppose an allocation $A$ is $\frac{c}{n}$-EF: then for all $i$ and $i'$, $v_i(A_i) \geq \frac{c}{n} v_i(A_{i'})$. For $i \le 2$, we have $v_i(A_3) \geq v_i(\{g_3\}) = k$, so
\[
v_i(A_i) \geq \frac{c}{n} v_i(A_3) \geq \frac{kc}{n} = \frac{c}{n}v_i(M)
\]
so $A$ is $c$-Prop.
\end{proof}

This resolves the $n>2$ case for all combinations of other parameters, so we will assume that $n=2$ for the remainder of the paper.

\section{Subadditive valuations}\label{sec:subadd}

In this section, we consider the deterministic setting for two players with subadditive valuations. In Section~\ref{sec:n=2-subadd-upper}, we use the Minimal Bundle Protocol from Section~\ref{sec:EF-n=2-submod-upper} to show that $c$-EF for $c \leq 1/2$ and $c$-Prop for $c \leq 2/3$ require only polynomial communication. This is the same protocol that yielded the PAS for EF with submodular valuations, but the communication cost analysis will be different. In Section~\ref{sec:n=2-subadd-lower}, we show that this is tight, by giving an exponential lower bound for $c$-EF and $c$-Prop when $c$ exceeds $1/2$ and $2/3$, respectively.

\subsection{Upper bounds}\label{sec:n=2-subadd-upper}

In this section, we prove that when players have subadditive valuations, the Minimal Bundle Protocol (Protocol~\ref{pro:minimal-upper}) can be used to solve \fd for $\frac{1}{2}$-EF and $\frac{2}{3}$-Prop with polynomial communication. In fact, we will show that if a satisfactory allocation is not found in step 1, there must exist a single item $g$ where $v_1(\{g\}) > v_1(M\backslash \{g\})$. This will imply that the only minimal bundle is $\{g\}$. Protocol~\ref{pro:minimal-upper} is restated here for the convenience of the reader.

In Section~\ref{sec:EF-n=2-submod-upper}, we proved correctness of this protocol for any setting, so it remains only to prove the communication cost bound for this setting. 

% kind of sketchy but oh well... I want the same number for the protocol when
% I restate it, and the restatable command did this well for theorems but not
% for algorithms/protocols. So I'm just going to be super hacky by saving the
% counter earlier and explicitly setting the algorithm counter to the
% counter for this protocol, and then setting it back
\newcounter{algorithm saved}
\setcounter{algorithm saved}{\value{algorithm}}
\setcounter{algorithm}{\value{pro:minimal-upper}}
\proMinimalUpper* % using "restatable" from the thm-restate package
\setcounter{algorithm}{\value{algorithm saved}}

Let $\alpha \in (0,1]$ be some constant. Let $\eta^P(\alpha)$ be the maximum $c\leq 1$ for which any allocation $A$ is guaranteed to be $c$-$P$, given $v_i(A_i) \geq \alpha v_i(A_{\overline{i}})$ for both $i$. For example, $\eta^{EF}(\alpha) = \alpha$. We will write $\eta^P(\alpha) = \eta(\alpha)$ and leave $P$ implicit. Lemma~\ref{lem:subadd-minimal-cost} is strongest for $\alpha = 1/2$, but we find it insightful to prove the theorem for any $\alpha \leq 1/2$.

Also, recall that Condition~\ref{cond:happy-with-one} is satisfied for $c$-Prop with subadditive valuations, for any $c$: for any allocation $A$, each player must be happy with at least one of $A$ and $\overline{A}$.

%%%%%%%%%%%%%%%%%%%%%%%%%%%%%%%%%%%%

\begin{lemma}\label{lem:subadd-minimal-cost}
For two players with subadditive valuations, $\alpha \in (0, 1/2]$, and $c = \eta(\alpha)$, Protocol~\ref{pro:minimal-upper} has communication cost at most
\[
2(m + v^{size})
\]
\end{lemma}

\begin{proof}
If the protocol terminates in step 1, a single allocation is communicated, which requires $m$ bits. Thus the claim is satisfied in this case. 

If the protocol does not terminate in step 1, the only communication happens in step 3. For a bundle $S$, $\bfc_{i}(S)$ is defined as the ratio of two values: $\mfrac{v_i(S)}{v_i(M\backslash S)}$ for EF, and $\mfrac{2v_i(S)}{v_i(M)}$ for Prop\footnote{Technically only one value is needed for Prop, since we can assume that $v_1(M) = 1$, so only $v_1(S)$ is needed. However, since two values are needed for EF, we ignore this.}. Thus communicating $\bfc_{1}(S^*(v_{1}))$ requires $2v^{size}$ bits. The only other information transmitted is the bundle $S^*(v_{1})$ and $\cals$. Communicating $S^*(v_{1})$ requires $m$ bits, and $\cals$ requires $|\cals | m$ bits. Thus the communication cost of the protocol is
\[
m(|\cals | + 1) + 2v^{size}
\]
It remains to show that if the protocol does not terminate in step 1, $|\cals| = 1$. 

By Condition~\ref{cond:happy-with-one}, for every allocation $A$, player 1 is happy with at least one of $A$ and $\overline{A}$. Thus $|\cals| \geq 1$, so let $S$ be a minimal bundle in $\cals$. Since player 1 is happy with $S$, we know that she is not happy with $M\backslash S$, or the protocol would have terminated in step 1. Suppose $\alpha v_1(S) \leq v_1(M\backslash S)$: then player 1 is $\eta(\alpha)$-happy with $M\backslash S$. Since $c = \eta(\alpha)$, this means player 1 is happy with $M\backslash S$, which is a contradiction. Therefore $\alpha v_1(S) > v_1(M\backslash S)$. This also implies that $v_1(S) > 0$.

Also, since $S$ is minimal, player 1 is not happy with $S \backslash \{g\}$ for all $g \in S$. Therefore player 1 is happy with $(M\backslash S)\cup \{g\}$ for all $g \in S$. By the same argument as above, we have $\alpha v_1(M\backslash S)\cup \{g\}) > v_1(S\backslash \{g\})$.

Since $v_i(S) > 0$, $S$ must be nonempty, so let $g$ be an arbitrary item in $S$. By subadditivity of $v_1$, we have 
\begin{align*}
v_1(M\backslash S) + v_1(\{g\}) \geq v_1\Big((M\backslash S) \cup\{g\}\Big)
\end{align*}
Similarly,
\begin{align*}
v_1(S \backslash \{g\}) + v_1(\{g\}) \geq&\  v_1(S)\\
v_1(S \backslash \{g\}) \geq&\  v_1(S) - v_1(\{g\})
\end{align*}
Therefore
\begin{align*}
v_1(M\backslash S) + v_1(\{g\}) \geq&\ v_1\Big((M\backslash S) \cup\{g\}\Big)\\
>&\ \frac{1}{\alpha}v_1(S\backslash\{g\})\\
\geq&\ \Big(\frac{1}{\alpha} - 1\Big)v_1(S\backslash\{g\}) + v_1(S) - v_1(\{g\})
\end{align*}
Since $v_1(S) > \frac{1}{\alpha}v_1(M\backslash S)$, we have
\begin{align*}
v_1(M\backslash S) + v_1(\{g\}) >&\ \Big(\frac{1}{\alpha}- 1\Big)v_1(S\backslash\{g\}) + \frac{1}{\alpha}v_1(M\backslash S) - v_1(\{g\})\\
2v_1(\{g\}) \geq&\ \Big(\frac{1}{\alpha}- 1\Big)v_1(S\backslash\{g\}) + 
		\Big(\frac{1}{\alpha}- 1\Big)v_1(M\backslash S)\\
v_1(\{g\}) \geq&\ \frac{1}{2}\Big(\frac{1}{\alpha}- 1\Big)\Big(v_1(S\backslash\{g\}) + v_1(M\backslash S)\Big)
\end{align*}
By subadditivity of $v_1$, we have
\begin{align*}
v_1(S\backslash\{g\}) + v_1(M\backslash S) \geq&\ v_1\Big((S \cup (M\backslash S)) \backslash \{g\}\Big)\\
=&\ v_1(M\backslash\{g\})
\end{align*}
Therefore,
\begin{align*}
v_1(\{g\}) >&\ \frac{1}{2}\Big(\frac{1}{\alpha}- 1\Big) v_1(M\backslash\{g\})\\
v_1(\{g\}) \geq&\ \alpha v_1(M\backslash\{g\})
\end{align*}
where the final step is due to $0 < \alpha \leq 1/2$.

Thus player 1 is $\eta(\alpha)$-happy with the bundle $\{g\}$ by definition. Since the protocol did not terminate in step 1, player 1 must not be happy with $M\backslash \{g\}$. Therefore player $i$ is happy with a bundle $S$ if and only if $g\in S$, and so the only minimal bundle is $\{g\}$.
\end{proof}

Theorem~\ref{thm:n=2-subadd-upper} is immediately implied by the combination of Lemma~\ref{lem:minimal-correct} (correctness) and Lemma~\ref{lem:subadd-minimal-cost} (communication cost).

\begin{theorem}\label{thm:n=2-subadd-upper}
For two players with subadditive valuations and $c = \eta(1/2)$, Protocol~\ref{pro:minimal-upper} has communication cost at at most $2(m + v^{size})$, and either returns a $c$-$P$ allocation or a $c^*$-$P$ allocation.
\end{theorem}

%%%%%%%%%%%%%%%%%%%%%%%%%%%%%%%%%%%%%%%%

Theorem~\ref{thm:n=2-subadd-upper} immediately implies the following result.

\begin{theorem}\label{thm:n=2-subadd-upper-complexity}
For two players with subadditive valuations, a property $P$, and any constant $c \leq \eta^P(1/2)$, there exists a deterministic protocol with communication cost $2(m + v^{size})$ which solves \fd. Formally,
\[
D_{subadd}(2,m,P, c) \leq 2(m + v^{size})
\]
for any $c \leq \eta^P(1/2)$.
\end{theorem}

\begin{proof}
Run Protocol~\ref{pro:minimal-upper} to either find a $\eta^P(1/2)$-$P$ allocation, or to find a $c'$-$P$ allocation where $c'$ is the best possible. If a $\eta^P(1/2)$-$P$ allocation exists, then a $c$-$P$ allocation exists, since $c \leq \eta^P(1/2)$. If a $c^*$-$P$ allocation is returned where $c^* < \eta^P(1/2)$, then by definition of $c^*$, a $c$-$P$ allocation exists if and only if $c^* \geq c$. \end{proof}

Theorem~\ref{thm:EF-n=2-subadd-upper} is a direct consequence of Theorem~\ref{thm:n=2-subadd-upper-complexity} since $\eta^{EF}(\alpha) = \alpha$, and Theorem~\ref{thm:prop-n=2-subadd-upper} requires only a short proof.

\begin{theorem}\label{thm:EF-n=2-subadd-upper}
For two players with subadditive valuations, $P =$ EF, and any constant $c \leq 1/2$, there exists a deterministic protocol with communication cost $2(m + v^{size})$ which solves \fd. Formally,
\[
D_{subadd}(2,m,\text{EF}, c) \leq 2(m + v^{size})
\]
for any $c \leq 1/2$.
\end{theorem}

\begin{theorem}\label{thm:prop-n=2-subadd-upper}
For two players with subadditive valuations, $P =$ Prop, and any constant $c \leq 2/3$, there exists a deterministic protocol with communication cost $2(m+v^{size})$ which solves \fd. Formally,
\[
D_{subadd}(2,m,\text{Prop}, c) \leq 2(m + v^{size})
\]
for any $c \le 2/3$.
\end{theorem}

\begin{proof}
By Theorem~\ref{thm:n=2-subadd-upper}, we need only show that $\eta^{Prop}(1/2) \geq 2/3$. Suppose $v_1(A_1) \geq \alpha v_1(A_2)$. Then $v_1(A_2) \leq \mfrac{1}{\alpha} v_1(A_1)$, and by subadditivity of $v_1$, we have
\begin{align*}
v_1(M) =&\ v_1(A_1 \cup A_2)\\
\leq&\ v_1(A_1) + v_1(A_2)\\
\leq&\ v_1(A_1) + \frac{1}{\alpha}v_1(A_1)\\
=&\ \frac{\alpha+1}{\alpha} v_1(A_1)
\end{align*}
Therefore $v_1(A_1) \geq \mfrac{2\alpha}{\alpha+1} \left(\mfrac{1}{2}v_1(M)\right)$, so $\eta^{Prop}(\alpha) \geq \mfrac{2\alpha}{\alpha+1}$. Therefore $\eta^{Prop}(1/2) \geq 2/3$.
\end{proof}

%%%%%%%%%%%%%%%%%%%%%%%%%%%%%%%%%%%%%%%%
Since $\eta^{Prop}(1/2)\geq 2/3$, any $\frac{1}{2}$-EF $A$ allocation is also $\frac{2}{3}$-Prop. However, a $c'$-EF allocation where $c'$ is the maximum possible EF approximation ratio does not necessarily achieve the maximum possible approximation ratio for Prop. Consider the case where $M = \{g_1, g_2\}$ and the players valuations are given by
\[ v_1(S) =
   \begin{cases}
      9 & \textnormal{if}\ \ S = M\\   
      7 & \textnormal{if}\ \ S = \{g_1\}\\
      2 & \textnormal{if}\ \ S = \{g_2\}\ 	
   \end{cases}
\hspace{.6 in}
v_2(S) =
   \begin{cases}
      4 & \textnormal{if}\ \ S = M\\   
      4 & \textnormal{if}\ \ S = \{g_1\}\\
      1 & \textnormal{if}\ \ S = \{g_2\} 	
   \end{cases}
\]
and $v_i(\emptyset) = 0$ for both $i$. There is no $\frac{1}{2}$-EF allocation or $\frac{2}{3}$-Prop allocation in this instance. The allocation achieving the maximum EF approximation ratio is $A = (\{g_2\}, \{g_1\})$ which is $\frac{2}{7}$-EF. On the other hand, the allocation achieving the maximum Prop approximation ratio is $\overline{A}$, which is $\frac{1}{2}$-Prop.

\subsection{Lower bounds}\label{sec:n=2-subadd-lower}

In this section, we show that $\frac{1}{2}$-EF and $\frac{2}{3}$-Prop are the best we can do deterministically for two players with subadditive valuations. We first prove that $c$-EF is hard for any $c > 1/2$, and then show that the same construction also proves hardness for $c$-Prop when $c > 2/3$.

We will use Lemma~\ref{lem:EF-equality-reduction}, which gives a standardized way to prove deterministic lower bounds for EF for two players. Recall that a list of allocations $\T = (T_1, T_2...)$ is defined where each $T_j = (T_{j1}, T_{j2})$ is an allocation giving each player $k$ items. Also, for every such allocation $A$, exactly one of $A$ and $\overline{A}$ appears in $\T$. All that is needed to complete the reduction is to show how to construct valuations $v_1, v_2$ such that the following conditions are satisfied:

\firstReductionCond*
\secondReductionCond*
\thirdReductionCond*

\begin{theorem}\label{thm:EF-n=2-subadd-lower}
For two players with subadditive valuations and any $c> 1/2$, any deterministic protocol which determines whether a $c$-EF allocation exists requires an exponential amount of communication. Formally,
\[
D_{subadd}(2,2k,\text{EF},c) \geq \frac{1}{2} \binom{2k}{k}
\]
for any $c > 1/2$.
\end{theorem}

\begin{proof}
Given bit strings of length $\mfrac{1}{2}\mbinom{2k}{k}$ for some integer $k$, define $M, N, (y_1, y_2)$, and $\T$ as in Lemma~\ref{lem:EF-equality-reduction}. We need only to construct subadditive valuations $v_1, v_2$ such that Conditions~\ref{reduction-cond1}, \ref{reduction-cond2}, and \ref{reduction-cond3} are met. We define each $v_i$ by
\[ 
v_i(S) =
   \begin{cases} 
      0 & \textnormal{if}\ \ |S| = 0\\        
      1 & \textnormal{if}\ \ 0 < |S| < k \\ 
      2 & \textnormal{if}\ \ k < |S| < 2k\\        
      3 & \textnormal{if}\ \ |S| = 2k\\           
      2 & \textnormal{if}\ \ \exists j\ \ S = T_{ji} \ \ \textnormal{where}\ \ y_{ij} = 1\\
      2 & \textnormal{if}\ \ \exists j\ \ S = T_{j\overline{i}} \ \ \textnormal{where}\ \ y_{ij} = 0\\   
      1 & \textnormal{if}\ \ \exists j\ \ S = T_{ji} \ \ \textnormal{where}\ \ y_{ij} = 0\\
      1 & \textnormal{if}\ \ \exists j\ \ S = T_{j\overline{i}} \ \ \textnormal{where}\ \ y_{ij} = 1\\     
   \end{cases}
\]
When $|S| = k$, $S$ falls under exactly one of the last four cases in the definition of $v_i$.

If $|S| < k$, we have $|\mss| > k$, so $v_i(S) \leq 1$ and $v_i(\mss) \geq 2$. Thus for any $c > 1/2$, $v_i(S) < c\cdot v_i(\mss)$, so Condition~\ref{reduction-cond1} is met. Suppose $y_{ij} = 1$ for some $i,j$: then $v_i(T_{j\overi}) = 1$ and $v_i(T_{ji}) = 2$, so again $v_i(S) < c\cdot v_i(\mss)$ for any $c > 1/2$. Suppose $y_{ij} = 0$ for some $i,j$: then similarly, $v_i(T_{ji}) = 1 < c\cdot 2 = c\cdot v_i(T_{j\overi})$ for any $c > 1/2$. Thus Condition~\ref{reduction-cond3} is satisfied as well.

It remains to show that $v_i$ is subadditive for both $i$. Specifically, we need to show that for any $S$ and $T$, $v_i(S) + v_i(T) \geq v_i(S \cup T)$. If either $S = \emptyset$ or $T = \emptyset$, this trivially holds, so suppose $|S| > 0$ and $|T| > 0$. We proceed by case analysis. 

Case 1: $|S\cup T| < 2k$. Then $v_i(S \cup T) \leq 2$.  Since $|S| > 0$ and $|T| > 0$, we have
\[
v_i(S) + v_i(T) \geq 1 + 1 \geq 2 \geq v_i(S \cup T)
\]

Case 2: $|S\cup T| = 2k$. Then $v_i(S\cup T) = 3$. Since $v_i(S) \geq 1$ and $v_i(T) \geq 1$, it remains to show that at least one of $v_i(S) \geq 2$ and $v_i(T) \geq 2$ is true. Since $S\cup T = M$ in this case, we have $M \backslash S \subseteq T$. Observe that under these valuations, for any allocation $A$ where $v_i(A_i) \leq 1$, we have $v_i(A_{\overi}) \geq 2$. Thus if $v_i(S) \leq 1$, then $v_i(M\backslash S) \geq 2$, so $v_i(T) \geq 2$. Since $v_i$ only takes on integer values in this proof, if $v_i(S) > 1$, we have $v_i(S) \geq 2$. Thus at least one of $v_i(S) \geq 2$ and $v_i(T) \geq 2$ is true, so the claim is satisfied in this case. Thus $v_i$ is subadditive for both $i$.
\end{proof}

To prove hardness for proportionality, we reduce from envy-freeness.

\begin{theorem}\label{thm:prop-n=2-subadd-lower}
For two players with subadditive valuations and any $c> 2/3$, any deterministic protocol which determines whether a $c$-Prop allocation exists requires an exponential amount of communication. Formally,
\[
D_{subadd}(2,2k,\text{Prop},c) \geq \frac{1}{2}\binom{2k}{k}
\]
for any $c > 2/3$.
\end{theorem}

\begin{proof}
We reduce from \equ. Given an input $(x_1, x_2)$, we define $v_i$ as in the proof of Theorem~\ref{thm:EF-n=2-subadd-lower}. By Theorem~\ref{thm:EF-n=2-subadd-lower}, for any $c' > 1/2$, at least $\mfrac{1}{2}\mbinom{2k}{k}$ communication is required to determine whether a $c'$-EF allocation exists under these valuations. We will show that under these valuations, for any $c>2/3$ and any $c' > 1/2$, an allocation $A$ is $c$-Prop if and only if is $c'$-EF: thus the lower bound of Theorem~\ref{thm:EF-n=2-subadd-lower} will apply to $c$-Prop for $c > 1/2$ as well.\footnote{It is actually sufficient to show that for any $c > 2/3$, there exists such a $c' > 1/2$, but we prove that this holds for any $c' > 1/2$.}

Suppose an allocation $A$ is $c'$-EF for some $c' > 1/2$: then $v_i(A_i) \geq c' v_i(A_{\overi}) > \mfrac{1}{2} v_i(A_{\overi})$. Under these valuations, for any allocation $A$ where $v_i(A_i) \leq 1$, we have $v_i(A_{\overi}) \geq 2$. Thus $v_i(A_i)$ must be strictly greater than 1. Since these valuations only take on integer values, this implies that $v_i(A_i) \geq 2$ for both $i$. Therefore
\[
v_i(A_i) \geq 2 \geq \frac{3}{2} = \frac{1}{2} v_i(M) \geq \frac{c}{2} v_i(M)
\] 
for every $c > 2/3$, so $A$ is $c$-Prop for every $c > 2/3$.

Now suppose that $A$ is $c$-Prop for some $c > 2/3$: then $v_i(A_i) \geq \mfrac{c}{2}v_i(M) = \mfrac{3c}{2} > 1$ for both $i$. Thus we again have $v_i(A_i) \geq 2$ for both $i$, since these valuations only take on integer values. This also implies that $|A_i | > 0$ for both $i$, which means that $|A_i | < 2k$ for both $i$. Therefore $v_i(A_{\overi}) \leq 2$ for both $i$, so we have
\[
v_i(A_i) \geq 2 \geq v_i(A_{\overi}) \geq c' v_i(A_{\overi})
\]
for any $c' > 1/2$. Therefore $A$ is $c'$-EF for every $c' > 1/2$.

%The proof of Theorem~\ref{thm:EF-n=2-subadd-lower} showed that for any $c' > 1/2$, a $c'$-EF allocation exists if and only if $(x_1, x_2)$ is a no-instance of \equ. Since the set of $c$-Prop and $c'$-EF allocations are identical when $c > 2/3$ and $c' > 1/2$, a $c$-Prop allocation exists if and only if $(x_1, x_2)$ is a no-instance of \equ, which completes the reduction.
\end{proof}

Theorems~\ref{thm:EF-n=2-subadd-lower} and \ref{thm:prop-n=2-subadd-lower} resolve the deterministic subadditive case. We now move on to general valuations, and give the last few results we need to complete Table~\ref{tbl:results}.

\section{General valuations}\label{sec:general}

This section covers the remaining settings for envy-freeness and proportionality. In Section~\ref{sec:prop-n=2-rand-lower}, we show that $c$-Prop is hard for general valuations for any $c > 0$, in both the randomized and deterministic settings. Section~\ref{sec:EF-n=2-general-lower} gives a similar lower bound for $c$-EF for any $c > 0$, but only for deterministic protocols. In Section~\ref{sec:n=2-rand-upper}, we show that there actually exists an efficient randomized protocol for $c$-EF for any $c \in [0,1]$. We also show that this protocol works for proportionality in the subadditive case, again for any $c \in [0,1]$. These results conclude our study of envy-freeness and proportionality. 

\subsection{Proportionality randomized lower bound}\label{sec:prop-n=2-rand-lower}

Recall that \disj on bit strings of length $\ell$ has randomized communication complexity $\Omega(\ell)$. (Lemma~\ref{lem:disjoint-lower}).

\begin{theorem}\label{thm:prop-n=2-rand-lower}
For two players with general valuations and any $c>0$, any randomized protocol which determines whether a $c$-Prop allocation exists requires an exponential amount of communication. Specifically
\[
R(2,2k,\text{Prop},c) \in \Omega \left( \binom{2k}{k} \right)
\]
for any $c > 0$.
\end{theorem}

\begin{proof}
We reduce from \disj. Given bit strings $x_1$ and $x_2$ of length $\mbinom{2k}{k}$, we construct an instance of \fd as follows. Let $N = [2]$ be the set of players, and let $M = [2k]$ be the set of items. Let $\T = (T_1, T_2...T_{|\T|})$ be an arbitrary ordering of all of the allocations which give each player $k$ items: for any allocation $A = (A_1, A_2)$ where $|A_1| = |A_2| = k$, there exists $j$ where $A_i = T_{ji}$ for both $i$. Both $A$ and $\overline{A}$ appear in $\T$. Note that $|\T| = \mbinom{2k}{k}$. Each player $i$'s valuation is defined by
\[ 
v_i(S) =
   \begin{cases}    
      0 & \textnormal{if}\ \ |S| < k\\     
      1 & \textnormal{if}\ \ |S| > k\\  
      1 & \textnormal{if}\ \ \exists j\ \ S = T_{ji} \ \ \textnormal{where}\ \ x_{ij} = 1\\
      0 & \textnormal{if}\ \ \exists j\ \ S = T_{ji} \ \ \textnormal{where}\ \ x_{ij} = 0
   \end{cases}
\]
Exactly one of the last two cases occur when $|S| = k$, and any such $j$ is unique.

Suppose that $(x_1, x_2)$ is a no-instance of \disj: then there exists $j$ where $x_{1j} = x_{2j} = 1$. Consider the allocation $T_j = (T_{j1}, T_{j2})$. Then for both $i$, $v_i(T_{ji}) = 1 = v_i(M) \geq \mfrac{c}{2} \cdot v_i(M)$, so the allocation $T_j$ satisfies $c$-Prop.

Suppose that $(x_1, x_2)$ is a yes-instance of \disj: then for every $j$, there exists $i$ where $x_{ij} = 0$. Suppose that a $c$-Prop allocation $A = (A_1, A_2)$ exists: then $v_i(A_i) \geq \mfrac{c}{2} \cdot v_i(M) > 0$ for both $i$. Suppose a player $i$ receives strictly more than $k$ items in $A_i$: then the other player receives strictly fewer than $k$ items, and has value zero, which is impossible. Thus $|A_1| = |A_2| = k$. Since $\T$ contains all of the allocations which give each player $k$ items, there must exist $j$ where $A_i = T_{ji}$ for both $i$. But that implies that $x_{1j} = x_{2j} = 1$, which is a contradiction. Therefore no allocation is $c$-Prop.
\end{proof}

This lower bound is actually much more general than just $c$-Prop. It holds for any imaginable fairness property (not just $c$-EF or $c$-Prop) where (1) player $i$ is always unhappy if $v_i(A_i) = 0$, even if $v_i(A_{\overi})$ is also 0, and (2) player $i$ is always happy if $v_i(A_i) = v_i(M)$. Both $c$-EF and $c$-Prop satisfy the first condition. The second condition is satisfied by $c$-Prop for any $c$, but $c$-EF violates this for every $c$: player $i$ is always happy if $v_i(A_i) = v_i(A_{\overi}) = 0$. We will see in Section~\ref{sec:n=2-rand-upper} that this leads to an efficient randomized protocol for $c$-EF, for any $c \in [0,1]$.

\subsection{Envy-freeness deterministic lower bound}\label{sec:EF-n=2-general-lower}

In this section we prove that for general valuations, $c$-EF is hard in the deterministic setting for any $c > 0$. We will use Lemma~\ref{lem:EF-equality-reduction}; recall that we need only show how to construct valuations that satisfy the following conditions:

\firstReductionCond*
\secondReductionCond*
\thirdReductionCond*

\begin{theorem}\label{thm:EF-n=2-general-lower}
For two players with general valuations and any $c>0$, any deterministic protocol which determines whether a $c$-EF allocation exists requires an exponential amount of communication. Specifically,
\[
D(2,2k,\text{EF},c) \geq \frac{1}{2} \binom{2k}{k}
\]
for any $c > 0$.
\end{theorem}

\begin{proof}
We use Lemma~\ref{lem:EF-equality-reduction}. Given bit strings of length $\mfrac{1}{2} \mbinom{2k}{k}$ for some integer $k$, define $M, N, (y_1, y_2)$, and $\T$ as in Lemma~\ref{lem:EF-equality-reduction}. We need only to construct valuations $v_1, v_2$ such that Conditions~\ref{reduction-cond1}, \ref{reduction-cond2}, and \ref{reduction-cond3} are met. We define each $v_i$ by\[ 
v_i(S) =
   \begin{cases}    
      0 & \textnormal{if}\ \ |S| < k\\     
      1 & \textnormal{if}\ \ |S| > k\\  
      1 & \textnormal{if}\ \ \exists j\ \ S = T_{ji} \ \ \textnormal{where}\ \ y_{ij} = 1\\
      1 & \textnormal{if}\ \ \exists j\ \ S = T_{j\overline{i}} \ \ \textnormal{where}\ \ y_{ij} = 0\\      
      0 & \textnormal{if}\ \ \exists j\ \ S = T_{ji} \ \ \textnormal{where}\ \ y_{ij} = 0\\
      0 & \textnormal{if}\ \ \exists j\ \ S = T_{j\overline{i}} \ \ \textnormal{where}\ \ y_{ij} = 1\\      
   \end{cases}
\]
Recall that for every allocation $A$ which gives each player $k$ items, $\T$ (as defined by Lemma~\ref{lem:EF-equality-reduction}) contains exactly one of $A$ and $\overline{A}$. Thus if $|S| = k$, $S$ falls under exactly one of the last four cases in the definition of $v_i$.

If $|S| < k$, we have $|\mss| > k$, so $v_i(S) = 0 < c = c\cdot v_i(\mss)$. This satisfies Condition~\ref{reduction-cond1}. Suppose $y_{ij} = 1$ for some $i,j$: then $v_i(T_{j\overi}) = 0 < c = c\cdot v_i(T_{ji})$, so Condition~\ref{reduction-cond2} is satisfied. Suppose $y_{ij} = 0$ for some $i,j$: then similarly, $v_i(T_{ji}) = 0 < c = c\cdot v_i(T_{j\overi})$. Thus Condition~\ref{reduction-cond3} is satisfied as well.
\end{proof}

\subsection{A randomized upper bound}\label{sec:n=2-rand-upper}

Although $c$-EF is hard for general valuations in the deterministic setting, it admits an efficient randomized protocol for any $c \leq 1$. Fundamentally, this is because the randomized communication complexity of \equ is constant, while its deterministic complexity is the length of the string. Our deterministic lower bound in Section~\ref{sec:EF-n=2-general-lower} was based on a reduction from \equ: in this section, we reduce \emph{to} \equ.

Our protocol will actually be much more general than just $c$-EF. For example, it will also work for $c$-Prop for subadditive valuations, for any $c \in [0,1]$. More generally, it solves \fd with two players when $(c,P)$ satisfies two conditions:

\condHappyWithOne*

\begin{condition}\label{rand-upper-cond2}
Whether player $i$ is happy does not depend on any valuation other than $v_i$.
\end{condition}

All of the fairness properties we consider satisfy Condition~\ref{rand-upper-cond2}. The $c$-EF property satisfies Condition~\ref{cond:happy-with-one} for any $c \leq 1$. As mentioned before, $c$-Prop satisfies this for any $c \leq 1$ for subadditive valuations.

Despite being hard in the deterministic setting, \equ admits an efficient randomized protocol (Lemma~\ref{lem:equality-upper}), as described in Section~\ref{sec:lower-bound-approach}. This protocol (let us call it $\Gamma_{EQ}$) enables the \fd randomized protocol that we present in this section.

The standard \equ problem is a decision problem, but \fd is a search problem: we must output a satisfactory allocation if one exists. The search version of \textsc{Equality} is to determine whether two bit strings are equal, and if they are not, to return an index where they differ.

 %%%%%%%%%%%%%%%%%%%%%%%%%%%%%%%%%%%%%%%
\begin{lemma}[\cite{Nisan94:Communication}]\label{lem:equality-search}
There exists a randomized protocol which solves the search version of \textsc{Equality} for two players and has communication cost $O(\log \ell)$, where $\ell$ is the length of the bit strings.
\end{lemma}

The protocol uses a binary search approach. The players first use $\Gamma_{EQ}$ to check if their strings are equal. If so, the protocol terminates. If not, the players split their strings into a left half and a right half. They again use $\Gamma_{EQ}$ to check if their left halves are equal: if they are not, the players recurse on the left half, otherwise they recurse on the right half. This process continues until players isolate a single bit which differs in their bit strings\footnote{The protocol described in~\cite{Nisan94:Communication} is actually slightly stronger: they find the most significant bit where the two strings differ. This is because they have a slightly different goal in that paper, for which finding any bit that differs is not sufficient.}.

Since $\Gamma_{EQ}$ is a randomized protocol, it may return an incorrect answer with probability up to 1/3 (say) each time it is run. If we use $\Gamma_{EQ}$ many times, as required by the above binary search argument, the probability $\Gamma_{EQ}$ returns the correct answer every time may be less than 2/3, which is unacceptable. This makes the protocol a sort of ``noisy binary" search. This can be done with total communication $O(\log \ell \log \log \ell)$ using a standard Chernoff bound argument, but~\cite{Feige94:Computing} shows how this can be done with total communication just $O(\log \ell)$. We refer to protocol from Lemma~\ref{lem:equality-search} as $\Gamma_{EQS}$.

 %%%%%%%%%%%%%%%%%%%%%%%%%%%%%%%%%%%%%%% 

We now present our randomized protocol (Protocol~\ref{pro:equality-upper}). Let $\T = (T_1, T_2...)$ be a list of every possible allocation (not just those with bundles of a fixed size) in an arbitrary order. Condition~\ref{rand-upper-cond2} is necessary for Protocol~\ref{pro:equality-upper} to be well-defined (step 1 in particular), but will not appear in the proof of Theorem~\ref{thm:equality-upper}.

The protocol uses a similar construction to the previous lower bounds in that players have exponentially long bit strings, with each index representing a possible allocation, and where $y_{ij} = 1$ if player $i$ is happy with $T_j$. Similarly to the \equ lower bounds, an index where $y_{1j} = y_{2j}$ implies the existence of a $c$-$P$ allocation: if $y_{1j} = y_{2j} = 1$, both players are happy with that allocation, and if $y_{1j} = y_{2j} = 0$, both players are happy with the reverse allocation by Condition~\ref{cond:happy-with-one}. This is made formal by the following theorem:

%An interesting observation is that since every possible allocation is listed in $\T$, there exists an index $j'$ where $\overline{T_j} = T_{j'}$, so $x_{1j'} = x_{2j'} = 1$. Thus our protocol does always (indirectly) determine if there is an index where both players have a 1-bit, but the fact that it is also allowed to report $j$ where $x_{1j} = x_{2j} = 0$ is enables a polynomial communication cost.
%\TODO{is the discussion above interesting/necessary?}

%%%%%%%%%%%%%%%%%%%%%%%%%%%%%%%%%%%%%%%
\begin{protocol}
\begin{center}
Private inputs: $v_1, v_2$\\
Public inputs: $P,c, \T$
\end{center}
\begin{enumerate}
\item Each player $i$ constructs a bit string $y_i$ as follows: for all $j$ where player $i$ is happy with $T_j$, player $i$ sets $y_{ij} = 1$. For all $j$ where player $i$ is unhappy with $T_j$, player i sets $y_{ij} = 0$.

\item Player 1 sets $x_1 = y_1$ and player 2 sets $x_2 = \overline{y_2}$.

\item The players run $\Gamma_{EQS}$ on $(x_1, x_2)$, which either returns an index $j$ where $x_{1j} \neq x_{2j}$, or determines that the two strings are equal.

\item If the two bit strings are equal, the players declare that no $c$-$P$ allocation exists.

\item If an index $j$ is returned where $x_{1j} = 1$ and $x_{2j} = 0$, the players declare that $T_j$ is a $c$-$P$ allocation.

\item If an index $j$ is returned where $x_{1j} = 0$ and $x_{2j} = 1$, the players declare that $\overline{T}_j$ is a $c$-$P$ allocation.
\end{enumerate}
\caption{Randomized protocol for two players to either find an $P$ allocation or determines that none exists, assuming $P$ satisfies Conditions~\ref{cond:happy-with-one} and \ref{rand-upper-cond2}.}
\label{pro:equality-upper}
\end{protocol}

%%%%%%%%%%%%%%%%%%%%%%%%%%%%%%%%%%%%%%%

\begin{theorem}\label{thm:equality-upper}
If $(c,P)$ satisfies Conditions~\ref{cond:happy-with-one} and \ref{rand-upper-cond2}, then Procotol~\ref{pro:equality-upper} either finds a $c$-$P$ allocation or shows that none exists, and uses communication $O(m)$.
\end{theorem}

\begin{proof} 
Suppose Protocol~\ref{pro:equality-upper} declares that no $c$-$P$ allocation exists: then $x_{1j} = x_{2j}$ for all $j$. This implies that $y_{1j} \neq y_{2j}$ for all $j$. Therefore whenever player 1 is happy with $T_j$, player 2 is unhappy with $T_j$, so no $c$-$P$ allocation exists.

Suppose Protocol~\ref{pro:equality-upper} returns an index $j$ where $x_{1j} \neq x_{2j}$. If $x_{1j} = 1$ and $x_{2j} = 0$, then $y_{1j} = y_{2j} = 1$. Thus both players are happy with $T_j$, so $T_j$ is $c$-$P$. If  $x_{1j} = 0$ and $x_{2j} = 1$, then $y_{1j} = y_{2j} = 0$, so neither player is happy with $T_j$. Then by Condition~\ref{cond:happy-with-one}, both players are happy with $\overline{T_j}$, so $\overline{T}_j$ is $c$-$P$. Therefore Protocol~\ref{pro:equality-upper} correctly finds a $c$-$P$ allocation or determines that none exist.

Since the total number of allocations is $O(2^m)$ when $n=2$, $x_1$ and $x_2$ have length $O(2^m)$. Thus $\Gamma_{EQS}$ has communication cost $O\big(\log (2^m)) = O(m)$. Since all other steps require no communication, Protocol~\ref{pro:equality-upper} uses communication $O(m)$.
\end{proof}

%%%%%%%%%%%%%%%%%%%%%%%%%%%%%%%%%%%%%%%

Theorem~\ref{thm:equality-upper} immediately implies the following two theorems:

\begin{theorem}\label{thm:EF-n=2-rand-upper}
For any $c \in [0,1]$, Protocol~\ref{pro:equality-upper} finds an $c$-EF allocation or shows that none exists, and has communication cost $O(m)$. Formally,
\[
R(2, m, EF, c) \in O(m)
\]
\end{theorem}

\begin{theorem}\label{thm:prop-n=2-rand-upper}
For subadditive valuations and any $c \in [0,1]$, Protocol~\ref{pro:equality-upper} finds an $c$-Prop allocation or shows that none exists, and has communication cost $O(m)$. Formally,
\[
R_{subadd}(2, m, Prop, c) \in O(m)
\]
\end{theorem}

Since $R_{submod}(n,m,P,c) \leq R_{subadd}(n,m,P,c) \leq R(n,m,P,c)$, this settles the randomized communication complexities for all settings with two players. The reader can verify that Table~\ref{tbl:results} is now complete.

\section{Maximin share}\label{sec:maximin}

Finally, we consider a different fairness property: maximin share. A player's maximin share (MMS) is the maximum value as she could guarantee herself if she gets to divide the items into $n$ bundles, but chooses last. An allocation $A$ is $c$-MMS for $c \in [0,1]$ if each player receives at least a $c$-fraction of her MMS. We use ``MMS" to refer to both each player's maxmin share and the fairness property itself. Formally,
\begin{definition}
An allocation $A$ is $c$-MMS if for every player $i$,
\[
v_i(A_i) \geq \max\limits_{A' = (A'_1...A'_n)}\  \min_{j \in [n]} v_i(A'_j)
\]
where $A'$ ranges over all possible allocations.
\end{definition}

In this section, we prove exponential lower bounds for MMS in two settings: for general valuations and any $c > 0$, and for submodular valuations when $c = 1$. Both lower bounds hold even for two players, and for randomized protocols. Both lower bounds will rely on reductions from \disj.

%%%%%%%%%%%%%%%%%%%%%%%%%%%%%%%%%%%%%%

\subsection{Lower bound for general valuations and any $c > 0$}

%\TODO{better to reduce from Prop here, or to choose prove from scratch? It might actually be simpler to just prove it from scratch using similar valuations to \ref{thm:prop-n=2-rand-lower}}

In this section we show that for general valuations, $c$-MMS is hard for any $c > 0$, even for randomized protocols and even if there are only two players. We will reduce from 1-Prop, which we know to be hard in this setting (randomized, $n=2$, general valuations) from Theorem~\ref{thm:prop-n=2-rand-lower}. We say that an allocation $A$ is over a set of items $M$ to mean that $A_1 \cup A_2 = M$. Also, we say that an allocation $A$ is $c$-Prop for valuations $v_1, v_2$ if $v_i(A_i) \geq \frac{1}{2} v_i(M)$ for both $i$. Since we will be reducing between two different \fd instances, we will be dealing with allocations over different sets of items and different sets of valuations.

\begin{theorem}\label{thm:MMS-general}
For two players with general valuations and any $c>0$, any randomized protocol which determines whether a $c$-MMS allocation exists requires an exponential amount of communication. Specifically,
\[
R(2,2k + 4,\text{MMS},c) \in \Omega\left( \binom{2k}{k} \right)
\]
for any $c > 0$.
\end{theorem}

\begin{proof}
Consider an arbitrary instance of \fd for two players with general valuations $v_1, v_2$, any $c > 0$, and some set of items $M$. We want to know whether there exists an allocation $A$ over $M$ which is 1-Prop for $v_1, v_2$. Let $\alpha_i = \mfrac{1}{2c}v_i(M)$: then $A$ is 1-Prop if and only if $v_i(A_i) \geq c\alpha_i$ for both $i$.

We will create a second instance of \fd as follows. Add four items $g_1, g_2, g_3, g_4$, let $X = \{g_1, g_2, g_3, g_4\}$, and define $M' = M \cup X$. Let $Y_1 = \{g_1, g_2\}, Y_2 =\{g_3, g_4\}, Z_1 = \{g_1, g_3\}$, and $Z_2 = \{g_2, g_4\}$. The set of players is the same. Define the following valuations $v_1'$ and $v_2'$ over $M'$:
\[ v_1'(S) =
   \begin{cases} 
      \alpha_1 & \textnormal{if}\ \ Y_1 \subseteq S\ \ \textnormal{or}\ \ Y_2\subseteq S\\
      \min(v_1(S \backslash X), c\alpha_1) & \textnormal{if}\ \ \{g_1, g_4\}\subseteq S
      		\ \ \textnormal{and}\ \ g_2, g_3 \not\in S\\
      0 & \textnormal{otherwise}
   \end{cases}
\]
\[ v_2'(S) =
   \begin{cases} 
      \alpha_2 & \textnormal{if}\ \ Z_1 \subseteq S\ \ \textnormal{or}\ \ Z_2 \subseteq S\\
      \min(v_2(S \backslash X), c\alpha_2) & \textnormal{if}\ \ \{g_2, g_3\}\subseteq S
      		\ \ \textnormal{and}\ \ g_1,g_4 \not\in S\\
      0 & \textnormal{otherwise}
   \end{cases}
\]
We first claim that each player $i$'s MMS is exactly $\alpha_i$. Since $c\leq 1$, we have $v_i'(A_i') \leq \alpha_i$ for all $i$ and for every allocation $A'$ over $M'$. Thus each player's MMS is at most $\alpha_i$. In the partition $(Y_1, Y_2 \cup M)$, player 1 has value $\alpha_1$ for both bundles, so player 1's MMS is at least $\alpha_1$. Similarly, player 2 has value $\alpha_2$ for both bundles in the partition $(Z_1, Z_2\cup M)$. Thus each player $i$'s MMS is exactly $\alpha_i$. 

Therefore an allocation $A'$ over $M'$ is $c$-MMS for $v_1', v_2'$ if and only if $v'_i(A'_i) \geq c\alpha_i$ for both $i$. It remains to show that there exists such an allocation $A'$ over $M'$ if and only if an there exists a 1-Prop allocation for $v_1, v_2$ over $M$.

Suppose $A$ is 1-Prop allocation over $M$ for $v_1, v_2$: then $v_i(A_i) \geq c\alpha_i$ for both $i$. Let $A' = (A_1 \cup X, A_2)$: then $v'_i(A'_i) \geq v_i(A_i) \geq c\alpha_i$, so $A'$ is $c$-MMS for $v_1', v_2'$ over $M'$.

Now suppose $A'$ is a $c$-MMS allocation for $v_1', v_2'$ over $M'$. Since $c>0$, we have $v'_i(A'_i) \geq c\alpha_i > 0$ for all $i$. For all $j$ and $j'$, we have $Y_j \cap Z_{j'} \neq \emptyset$. Thus if player 1 receives $Y_1$ or $Y_2$: then player 2 player 2 cannot receive $Z_1$ or $Z_2$. Furthermore, player 2 cannot receives $\{g_2, g_3\}$, so $v_2'(A'_2) = 0$, which is a contradiction. Therefore player 1 cannot receive either $Y_1$ or $Y_2$. Similarly, if player 2 receives $Z_1$ and $Z_2$, player 1 will have value 0. Thus player 2 does not receive $Z_1$ or $Z_2$.

Therefore $v_1(A_1) = 0$ unless $\{g_1, g_4\} \subseteq A_1$, and $v_2(A_2) = 0$ unless $\{g_2, g_3\} \subseteq A_2$. Therefore $\{g_1, g_4\} \subseteq A_1$ and $\{g_2, g_3\} \subseteq A_2$. Thus $v'_i(A'_i) = \min(v_i(A'_i \backslash X), c\alpha_i)$ for both $i$. Since $v'_i(A'_i) \geq c\alpha_i$ for all $i$, we have $v_i(A'_i \backslash X) \geq c\alpha_i$ for both $i$.

Define an allocation $A$ where $A_i = A'_i \backslash X$. It is clear that $A$ is an allocation over $M$. Then $v_i(A_i) \geq c\alpha_i$ for both $i$, so $A$ is a 1-Prop allocation for $v_1, v_2$ over $M$.

Therefore there exists a $c$-MMS allocation for $v_1', v_2'$ over $M'$ if and only if there exists 1-Prop allocation over for $v_1, v_2$ over $M$. This completes the reduction, and shows that for any $c > 0$ and any number of items $m$,
\[
R(2,m+4,\text{MMS},c) \geq R(2,m,\text{Prop},1)
\]
Therefore by Theorem~\ref{thm:prop-n=2-rand-lower}, we have $R(2,2k + 4,\text{MMS},c) \in \Omega\left( \mbinom{2k}{k} \right)$.
\end{proof}

%%%%%%%%%%%%%%%%%%%%%%%%%%%%%%%%%%%%%%%
\subsection{Lower bound for submodular valuations and $c=1$}

We now show that even for two players with submodular valuations, 1-MMS is hard. This does not hold for $c$-MMS for any $c$: in fact, a $\frac{1}{3}$-MMS is guaranteed to exist for submodular valuations~\cite{Ghodsi17:Fair}.

\begin{theorem}
For two players with submodular valuations, any randomized protocol which determines whether a $1$-MMS allocation exists requires an exponential amount of communication. Specifically
\[
R(2,2k,\text{MMS},1) \in \Omega \left( \binom{2k}{k} \right)
\]
\end{theorem}

\begin{proof}
We reduce from \disj. Given bit strings $x_1$ and $x_2$ of length $\mbinom{2k}{k}$, we construct an instance of \fd as follows. Let $N = [2]$ be the set of players, and let $M = [2k]$ be the set of items. We define $Y_1 = \{1...k\}, Z_1 = \{k+1...2k\}, Y_2 = \{2...k+1\}$, and $Z_2 = \{1\} \cup \{k+2...2k\}$. 

Let $\T = (T_1, T_2...T_{|\T|})$ be an arbitrary ordering of all of the allocations which give each player $k$ items: for any allocation $A = (A_1, A_2)$ where $|A_1| = |A_2| = k$, there exists $j$ where $A_i = T_{ji}$ for both $i$. Note that $|\T| = \mbinom{2k}{k}$. One exception: none of $(Y_1, Z_1), (Z_1, Y_1), (Y_2, Z_2)$, or $(Z_2, Y_2)$ appear in $\T$.

Player $i$'s valuation is given by:
\[ v_i(S) =
   \begin{cases}  
      3|S| & \textnormal{if}\ \ |S| < k\\
      3k & \textnormal{if}\ \ |S| > k\\
      3k & \textnormal{if}\ \ S = Y_i \ \ \textnormal{or}\ \ S = Z_i\\   
      3k - 1 & \textnormal{if}\ \ S = Y_{\overline{i}}
      		 \ \ \textnormal{or}\ \ S = Z_{\overline{i}}\\   
      3k & \textnormal{if}\ \ \exists j\ \ S = T_{ji} \ \ \textnormal{where}\ \ x_{ij} = 1\\
      3k-1 & \textnormal{if}\ \ \exists j\ \ S = T_{ji} \ \ \textnormal{where}\ \ x_{ij} = 0
   \end{cases}
\]
These valuations are submodular by the same argument as in the proof of Theorem~\ref{thm:EF-n=2-submod-lower}. Observe that when $|S| = k$, exactly one of the last four cases occur.

Since $v_i(S) \leq 3k$ for all $S$, player $i$'s MMS is at most $3k$. For both $i$, $(Y_i, Z_i)$ is a valid allocation. Furthermore, player $i$ has value $3k$ for both bundles in that allocation. Thus each player $i$'s MMS is at least $3k$, so both players' MMS are exactly $3k$.

Suppose that $(x_1, x_2)$ is a no-instance of \disj: then there exists $j$ where $x_{1j} = x_{2j} = 1$. Consider the allocation $T_j = (T_{j1}, T_{j2})$. Then for both $i$, $v_i(T_{ji}) = 3k$, so the allocation $T_j$ satisfies $1$-MMS.

Suppose that $(x_1, x_2)$ is a yes-instance of \disj: then for every $j$, there exists $i$ where $x_{ij} = 0$. Suppose a 1-MMS allocation $A$ exists. We first claim that for both $i$,  $A\neq (Y_i, Z_i)$ and $A\neq (Z_i, Y_i)$ for both $i$. This is because player $\overline{i}$ will have value $3k-1$, which is less than her MMS. Suppose there is a player $i$ where $|A_i | < k$: then $v_i(A_i) < 3k$, which is impossible. Thus $|A_1 | = |A_2 | = k$.

Therefore there exists $j$ where $A = T_j$. But since $(x_1, x_2)$ is a yes-instance of \disj, there exists $i$ where $x_{ij} = 0$, so $v_i(T_{ji}) = 3k - 1$. This is a contradiction, so no 1-MMS allocation exists.
\end{proof}

\section{Conclusion}\label{sec:conclusion}

In this paper, we proposed a simple model for the communication complexity of fair division, and solved it completely, for every combination of five parameters: number of players, valuation class, fairness property $P$, constant $c$, and deterministic vs. randomized complexity. 

More broadly, communication complexity is an example of topic that has been well-studied in algorithmic game theory but not in fair division, despite having a natural fair division analog. We wonder if there are other such topics.

\section*{Acknowledgements}

This research was supported in part by NSF grant CCF-1524062, a Google Faculty Research Award, and a Guggenheim Fellowship.

% Bibliography
\bibliographystyle{plain}
\bibliography{refs}

\end{document}